\documentclass[journal]{IEEEtran}
\usepackage{graphicx,amsmath,amssymb}
\usepackage{cite}
\usepackage{subfigure}
\usepackage{citesort}
\usepackage{fancyhdr}
\usepackage{mdwmath}
\usepackage{mdwtab}
\usepackage{balance}
\usepackage{xcolor}
\usepackage{bm}
\usepackage{amsthm}
\usepackage{algorithm}
\usepackage{algorithmic}
\usepackage{multirow}
\usepackage{flafter}
\usepackage{epstopdf}
\newtheorem{remark}{Remark}

\newtheorem{theorem}{Theorem}

\newtheorem{lemma}{Lemma}

\newtheorem{corollary}{Corollary}

\newtheorem{assumption}{Assumption}

\begin{document}
\title{I/Q Imbalance Aware Nonlinear Wireless-Powered Relaying of B5G Networks: Security and Reliability Analysis}

\author{Xingwang~Li,~\IEEEmembership{Senior Member,~IEEE,}
        Mengyan~Huang,~\IEEEmembership{Student Member,~IEEE,}
         Yuanwei~Liu,~\IEEEmembership{Senior Member,~IEEE,}
         Varun~G~Menon,~\IEEEmembership{Senior Member,~IEEE,}
         Anand~Paul,~\IEEEmembership{Senior Member,~IEEE,}
        and Zhiguo~Ding,~\IEEEmembership{Fellow,~IEEE}
\thanks{X. Li are M. Huang are with the School of Physics and Electronic Information Engineering, Henan Polytechnic University, Jiaozuo, China (email:lixingwangbupt@gmail.com, huangmengyan66@163.com).}
\thanks{Y. Liu is with the School of Electronic Engineering and Computer Science, Queen Mary University of London, London, UK, London, UK (email:yuanwei.liu@qmul.ac.uk).}
\thanks{V.~G.~Menon is with the department of Computer Science and Engineering, SCMS School of Engineering and Technology, Ernakulam, India (e-mail: varunmenon@ieee.org).}
\thanks{A.~Paul is with the School of Computer Science and Engineering, Kyungpook National University,  South Korea (e-mail: paul.editor@gmail.com).}
\thanks{Z. Ding is with the School of Electrical and Electronic Engineering, The University of Manchester, Manchester, UK (email: zhiguo.ding@manchester.ac.uk).}
}

\maketitle
\nocite{Li2018}
\begin{abstract}
Physical layer security is known as a promising paradigm to ensure security for the beyond 5G (B5G) networks in the presence of eavesdroppers. In this paper,
we elaborate on a tractable analysis framework to evaluate the reliability and security of wireless-powered decode-and-forward (DF) multi-relay networks. The nonlinear energy harvesters, in-phase and quadrature-phase imbalance (IQI) and channel estimation errors (CEEs) are taken into account in the considered system. To further improve the secure performance, two relay selection strategies are presented: \emph{1) suboptimal relay selection (SRS); 2) optimal relay selection (ORS)}. Specifically, exact analytical expressions for the outage probability (OP) and the intercept probability (IP) are derived in closed-form. For the IP, we consider that the eavesdropper can wiretap the signal from the source or the relay. In order to obtain more useful insights, we carry out the asymptotic analysis and diversity orders for the OP in the high signal-to-noise ratio (SNR) regime under non-ideal and ideal conditions. Numerical results show that: 1) Although the mismatches of amplitude/phase of transmitter (TX)/receiver (RX) limit the OP performance, it can enhance IP performance; 2) Large number of relays yields better OP performance; 3) There are error floors for the OP because of the CEEs; 4) There is a trade-off for the OP and IO to obtain the balance between reliability and security.
\end{abstract}


\section{Introduction}
The goals of the fifth generation (5G) and beyond (B5G) wireless networks can provide reliable communication among almost all aspects of life through the network with higher date rate, lower latency and ubiquitous connectivity \cite{8977557}. The security has been identified as an vital factor for wireless communication systems, which has triggered enormous interests from both academia and industry \cite{1,2}. However, due to the broadcast characteristics of wireless communication, it is difficult to ensure secure communication for the wireless networks without being eavesdropped by un-authored receivers. The conventional methods to ensure the security of wireless communication are to use encryption algorithms, which impose extra computational overhead and system complexity \cite{8879726}. In addition, with the rapid development of chip and computer technologies, conventional encryption technologies can not provide perfect security.

As an alternative way to ensure security, physical layer security (PLS) has sparked a great deal of research interests \cite{1}. The basic principle of PLS is to exploit the inherent randomness of fading channels to resist the information to be extracted by eavesdroppers \cite{2}.
Recently, there are a great of research works investigated the PLS under various fading channels, e.g. see \cite{5,4,6,7} and the references therein. In \cite{5}, a secure transmit-beamforming of the multiple-input multiple-output (MIMO) systems over Rayleigh fading channels was designed, in which the maximal ratio combing (MRC) receivers were adopted to maximize the signal-to-noise ratio (SNR) at main receiver. The authors of \cite{4} investigated the PLS of artificial noise aided MIMO systems over Rician fading channels. Meanwhile, the secure performance of the classic wiretap model was discussed over the generalized Gamma fading channels and the analytical expressions of the strictly positive secrecy capacity (SPSC) probability and the lower bound for the secrecy outage probability (SOP) were derived \cite{6}. On the other hand, the authors investigated the secrecy performance over $\kappa-\mu$ shadowed fading channels of classic Wyner's wiretap model and the approximate expressions on the lower bound for the SOP in the high SNR region and SPSC probability with the aid of a moment matching method have been obtained \cite{7}.

In actual situations, it is difficult to have direct links between the sources and the destinations due to shadow fading and/or obstacle, so it is indispensable to use the relay to complete the communication \cite{8,9}. In light of this fact, relaying assisted transmission has been identified as one of the key technologies in the current and future wireless cooperative networks \cite{10,11}. The signals can be decoded and transferred from the source to the destination by using low cost and low power consumption relay nodes. In general, there are two basic relay protocols: i) amplify-and-forward (AF) \cite{14,15,16}, and ii) decode-and-forward (DF) \cite{17,18}. In \cite{14}, the authors studied the ergodic capacity (EC) performance of fixed-gain AF dual-hop (DH) networks and derived two analytical expressions on the EC bound. Extending to multi-hop networks, the authors of \cite{15} derived the EC, outage probability (OP) and average symbol error probability by the generalized transformed characteristic function approach. In \cite{16}, the authors investigated the performance of a multi-hop AF communication network over Nakagami-0.5 channels and the closed-form analytical approximate expressions for the OP, ASEP and EC were obtained.  To maximize confidentiality, the authors in \cite{17} investigated the secure performance of multiple DF relay systems.

When deploying multiple relays in the systems, it will incur extra inter-relay (IR) interference and energy consumption. To avoid this problem, relay selection (RS) has been proposed \cite{19}. Among the various RS schemes, optimal relay selection (ORS), suboptimal relay selection (SRS) and MRC are the most prevalent ones \cite{23,24,25}. The pioneering work of the ORS scheme has been proposed by Bletsas according to selecting the relay with the largest instantaneous end-to-end SNR \cite{22}. Based on the ORS, the authors in \cite{23} investigated the symbol error rate (SER) of AF relay systems. To reduce the requirement of channel knowledge, the authors proposed a SRS scheme that the optimal relay is selected according to the link either source-relay or relay-destination \cite{24}. Cognitive radio inspired cooperative relay systems was introduced, and the secure outage performance was studied over independent and non-identically distributed Nakagami-$m$ fading channels. In \cite{25}, the authors compared the secrecy outage performance of cognitive radio networks for ORS, SRS with MRC schemes over Nakagami-$m$ fading channels.

Although the performance of wireless cooperative networks can be improved by appropriate relay protocols and RS scheme, the operation of wireless communication system is constrained by power shortages of their wireless devices. This happens that in some cases the nodes are deployed in the remote or power limited areas \cite{27}. In light of this context,
some energy harvesting (EH) techniques have been proposed to prolong the life of the batteries of such wireless transmission devices \cite{28,29}. Among the various EH techniques, radio frequency (RF) enabled simultaneous wireless information and power transfer (SWIPT) attracts a lot of attentions because it can overcome the limitations of some other renewable energy resources such as solar energy, wind energy and magnetic induction that can only be used in some specific circumstances \cite{30}. In addition, RF signals are ubiquitous in electromagnetic waves, and EH in RF is green, safe, controllable and reliable \cite{32}. There are usually two common protocols for SWIPT systems: i) time-switching (TS) and ii) power-splitting (PS) \cite{35,37,31}. For TS, the authors in \cite{35} investigated the outage performance of SWIPT-assisted non-orthogonal multiple access (NOMA) relay systems over Weibull fading channels. 
Considering PS protocol, the secure performance of two-way relaying systems was researched through a joint-optimization solution over geometric programming and binary particle swarm optimization \cite{37}. Additionally, a large-scale RF-EH technique with PS protocol was adopted, and the OP performance and average harvested energy were analyzed \cite{31}.

The aforementioned studies are based on the assumption of ideal hardware components and perfect channel state information (CSI), which is unrealistic in practical communication systems. In practice, due to component mismatch and manufacturing non-idealities, these monolithic architectures inevitably have defects associated with the RF front-ends, thereby limiting the overall system performance \cite{8879698}. A typical example of these impairments is the in-phase and quadrature-phase imbalance (IQI), which refers to the mismatches of amplitude and phase between I and Q branches of the transceiver. This will result in incomplete image suppression and ultimately lead to degradation of the performance for the total communication system \cite{39}. Ideally, the I and Q branches of the mixer have an amplitude of 0 and a phase shift of $90^\circ $, providing an infinitely attenuated image band; however, in practice, the transceiver is susceptible to some analog front-end damage, and these damages introduce errors in the phase shift resulting in amplitude mismatch between the I and Q branches, thereby damaging the down-converted signal constellation, thereby increasing the corresponding error \cite{38}. Motivated by the above practical concern, several research contributions have studied the systems secure performance in the presence of IQI \cite{41,9032127,XingwangLI}. Under the assumption of uncorrelation between channel of each subcarrier and its image, Ozdemir \emph{et al.} in \cite{41} derived an exact expression for the SINR of OFDM systems with IQI at transceivers. The authors analyzed the impact of joint IQI on the security and reliability of cooperative NOMA for IoT Networks \cite{9032127}. Considering backscatter communication, Li \emph{et al.} in \cite{XingwangLI} derived analytical expressions for OP and the intercept probability (IP) of ambient backscatter NOMA systems under IQI.
On the other hand, imperfect CSI (ICSI) may be existed due to the presence of channel estimation errors (CEEs) and feedback delay. Therefore, it is of great practical significance to study the impact of ICSI and IQI on the security performance of cooperative networks.

\subsection{Motivation and Contribution}
Motivated by the above discussion, we study the reliability and security of cooperative multi-relay networks in the presence of nonlinear energy harvesters, ICSI and IQI. Under these imperfect conditions, three RS schemes, random relay selection (RRS), SRS, ORS are considered. Specifically, we derive the analytical expressions for the OP and IP. For the security, the direct transmission and cooperative transmission through relay are considered. In this study, we assume that the source and relay nodes of the networks are configured with nonlinear energy harvesters to harvest energy from the nearby power beacon under different saturation thresholds. This is reasonable in some applications, such as internet-of-things (IoT), mesh networks and Ad Hoc networks, etc. The main contributions of this paper are summarized as follows:
\begin{itemize}
  \item Considering IQI and CEEs, we propose three representative RS schemes, namely RRS, SRS and ORS. RRS is considered as a benchmark for the purpose of comparison. In SRS, the optimal relay is selected according to the channel conditions either the $S-{R_m}$ or the ${R_m}-D$. In ORS, the optimal relay is selected according to the link qualities both the $S-{R_m}$ and the ${R_m}-D$. The major difference between our work and \cite{43} is that to study the effects of IQI caused by the mismatches of amplitude and phase between I and Q branch.

  \item Different from the most existing research works, we use a more realistic nonlinear EH model due to the nonlinearity of the electronic devices \cite{49,7264986Schober2015CL}. We have the assumption that nonlinear energy harvesters are equipped at source and relays, which can harvest energy from the nearby power beacon.
  \item For the reliability, we derived the exact analytical expressions for the OP of the proposed system for the three RS schemes. For the security, we consider two typical scenarios that direct transmission and cooperative transmission, the exact closed-form analytical expressions of the IP for the two scenarios are derived.\footnote{In some cases, the eavesdropper can simultaneously receive signals from both source and relays. Our work can be easily extended to these cases by combining the received signals from source and relays using the selection combining or MRC.}
  \item To obtain more insights, we derived the asymptotic analytical expressions and diversity orders for the OP of the three RS schemes under non-ideal conditions. It reveals that there are error floors for the OP due to the non-zero CEEs, and the OP performance is limited by the IQI parameters.


\end{itemize}

\begin{figure}[!t]
\setlength{\abovecaptionskip}{0pt}
\centering
\includegraphics [width=3.6in]{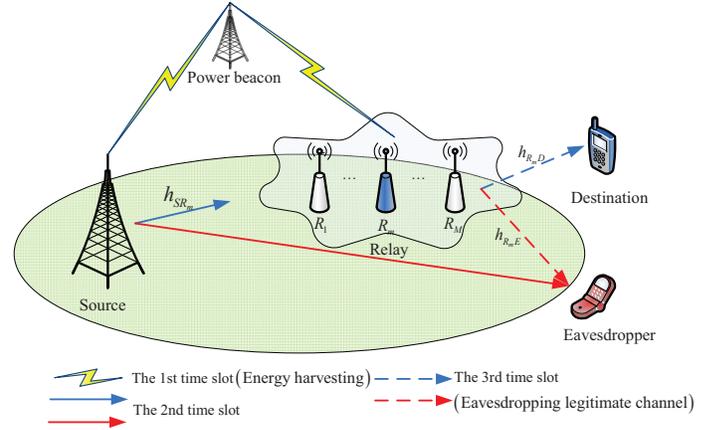}
\caption{System Model}
\label{fig1}
\end{figure}

\subsection{Organization and Notations}
The rest of the paper is organized as follows. In section II, we present a brief introduction of the considered system model. In section III, the reliability of the considered system for the three RS schemes is studied in terms of OP, while the security is analyzed through deriving expressions of the IP. In section IV, some numerical results are provided to verify the correctness of our analysis. Finally, we present a conclusion of this paper in Section V.

We use $\left|  \cdot  \right|$ to define absolute value. The notations ${\rm E}{\{ \cdot\} }$ and $\triangleq$ denote the expectation and definition operations, respectively. $e \sim \mathcal{CN}\left( {\mu ,{\sigma ^2}} \right)$ defines a complex Gaussian distribution with a mean of $\mu$ and a variance of ${\sigma ^2}$. $\Pr \left\{  \cdot  \right\}$ represents the probability and ${{\rm K}_v}\left(  \cdot  \right)$ denotes the $v$-th order modified Bessel function of the second kind. The probability density function (PDF) and cumulative distribution function (CDF) are expressed by ${f_X}\left(  \cdot  \right)$ and ${F_X}\left(  \cdot  \right)$, respectively. Finally, $\log_2 \left(  \cdot  \right)$ is the logarithm.

\section{System Model}\label{sec2}
We consider a DF multi-relay system as shown in Fig. 1, which deploys one power beacon $B$, one source $S$, $M$  relays ${R_m},m \in \left\{ {1,2, \cdots ,M} \right\}$, one destination $D$, one eavesdropper $E$ and all nodes equipped with a single antenna. All nodes are operate in half-duplex (HD) mode. In order to improve the secure performance of the considered system, the SRS and ORS schemes are used to select the optimal relay among the $M$ relays, while RRS scheme is presented as a benchmark. The source and all relay nodes are energy-constrained and can harvest energy from nearby $B$ according to the TS protocol. It is considered that there is no direct link of $S \to D$ due to the blockage or heavy shadowing.

In fact, it is difficult to obtain perfect CSI in the communication process because of the CEEs, and the most common method is to estimate the channel using the training sequence. In this study, the linear minimum mean square error (LMMSE) is adopted here. Thus, the channel can be modeled as
\begin{equation}
\label{1}
{h_j} = {\hat h_j} + {e_j},
\end{equation}
where ${\hat h_j},j \in \left\{ {S{R_m}, SE, {R_m}D,{R_m}E} \right\},\left( {1 \le m \le M} \right)$ is estimated channel of the real channel ${h_j}$, ${e_j} \sim \mathcal{CN}\left( {0,\sigma _{{e_j}}^2} \right)$ is the CEE, where $\sigma _{{e_j}}^2$ is the variance of estimation error and which is considered in two representative channel estimation models: 1) It is a non-negative fixed constant; 2) It is a function of transmit average SNR and can be modeled as $\sigma _{{e_j}}^2 = {{{\Omega _j}} \mathord{\left/
 {\vphantom {{{\Omega _j}} {\left( {1 + \delta {\rho _j}{\Omega _j}} \right)}}} \right.
 \kern-\nulldelimiterspace} {\left( {1 + \delta {\rho _j}{\Omega _j}} \right)}}$, where $\delta $ is the channel estimation quality parameter that indicates the power consumption of the training pilot to obtain CSI; ${\Omega _j}$ and ${\rho _j}$ are the variance of channel gain and transmit average SNR, respectively \cite{44}. We assume that all communication links are subject to Rayleigh fading and path loss \cite{45}.

Typically, IQI is modeled as the phase and/or amplitude imbalance between transceiver I and Q signal paths. As depicts in \cite{46,47}, the asymmetrical IQI model can be considered, where I branch and Q branch are assumed to be ideal and errors, respectively. Here, both transmitter (TX) and receiver (RX) are subject to IQI, in which case the transmitted baseband signals can be expressed as
\begin{equation}
\label{2}
{x_{IQI}} = {\mu _{t/r}}{y_j} + {v_{t/r}}y_j^*,
\end{equation}
where ${y_j} = \sqrt {P_{S/R}} {x_j}$ is the baseband signal that is transmitted under the conditions of non-ideal I/Q matching with ${\rm E}\left\{ {{{\left| {{x_j}} \right|}^2}} \right\} = 1$, ${x_j}$ is the transmit signal of the TX, ${P_S}$ and ${P_{{R_m}}}$ are the transmit power at $S$ and ${R_m}$, respectively; The IQI coefficients are given by ${\mu _t} = \frac{1}{2}\left( {1 + {\xi _t}\exp \left( {j{\phi _t}} \right)} \right)$, ${v_t} = \frac{1}{2}\left( {1 - {\xi _t}\exp \left( { - j{\phi _t}} \right)} \right)$, ${\mu _r} = \frac{1}{2}\left( {1 + {\xi _r}\exp \left( { - j{\phi _r}} \right)} \right)$, ${v_r} = \frac{1}{2}\left( {1 - {\xi _r}\exp \left( {j{\phi _r}} \right)} \right)$, where ${\xi _{{t \mathord{\left/ {\vphantom {t r}} \right. \kern-\nulldelimiterspace} r}}}$ and ${\phi _{{t \mathord{\left/ {\vphantom {t r}} \right. \kern-\nulldelimiterspace} r}}}$ denote the amplitude and phase mismatch at TX and RX, respectively \cite{Xingwang2019Electronics}. For ideal conditions, the parameters are set to ${\xi _{{t \mathord{\left/
 {\vphantom {t r}} \right.
 \kern-\nulldelimiterspace} r}}} = 1$ and ${\phi _{{t \mathord{\left/
 {\vphantom {t r}} \right.
 \kern-\nulldelimiterspace} r}}} = 0^\circ $ \cite{48}.

 The entire data transmission is completed in three phases: 1) $S$ and relays harvest energy from nearby power beacon $B$; 2) $S$ transmits its own signals to ${R_m}$ and $E$; 3) ${R_m}$ decodes and forwards the signals to $D$ and $E$.

{\emph{\it 1) The first phase:}} In this phase, $S$ and $R_m$ are equipped with nonliear harvesters can reap energy from $B$. The harvested energy at $S$ is
\begin{equation}
\label{3}
{E_S} = {\varsigma _1}{P_B}{\left| {{h_{BS}}} \right|^2}\alpha T,
\end{equation}
where ${\varsigma _1} \in (0, 1)$ is the energy converse coefficient of harvester at $S$; ${P_B}$ is the transmitted power at $B$; ${h_{BS}}$ is the channel coefficient between $B$ and $S$; $\alpha $ is the time allocation factor for EH, and $T$ is the block transmission duration. The harvested energy ${E_S}$ is used for information transmission in the second phase. The transmit power ${P_S}$ can be expressed as follows in the case of the nonlinear energy harvester \cite{49}
\begin{equation}
\label{4}
{P_S} = \left\{ \begin{array}{l}
\frac{{2\alpha {\varsigma _1}{P_B}}}{{1 - \alpha }}{\left| {{h_{BS}}} \right|^2},{\rm{  if }}{P_B}{\left| {{h_{BS}}} \right|^2} \le {\Gamma _1}{\rm{ }}\\
\frac{{2\alpha {\varsigma _1}}}{{1 - \alpha }}{\Gamma _1},\;\;\;\;\;\;\;\;\;\;{\rm{           if }}{P_B}{\left| {{h_{BS}}} \right|^2} > {\Gamma _1}{\rm{ }}
\end{array} \right.,
\end{equation}
where ${\Gamma _1}$ is the saturation threshold of the harvester at $S$.

Similarly, the energy harvested at ${R_m}$ can be expressed as
\begin{equation}
\label{5}
{E_{{R_m}}} = {\varsigma _2}{P_B}{\left| {{h_{B{R_m}}}} \right|^2}\alpha T,
\end{equation}
where ${\varsigma _2} \in (0, 1) $ is the energy conversion coefficient of harvester at ${R_m}$, and ${h_{B{R_m}}}$ is the channel coefficient from $B$ to $R_m$. The harvested energy $E_{{R_m}}$ at relays is used for the information transmission in the third phase. In the presence of the nonlinear energy harvester, the transmit power at relay is given as follows
\begin{equation}
\label{6}
{P_{{R_m}}} = \left\{ \begin{array}{l}
\frac{{2\alpha {\varsigma _2}{P_B}}}{{1 - \alpha }}{\left| {{h_{B{R_m}}}} \right|^2},{\rm{  if }}{P_B}{\left| {{h_{B{R_m}}}} \right|^2} \le {\Gamma _2}{\rm{ }}\\
\frac{{2\alpha {\varsigma _2}}}{{1 - \alpha }}{\Gamma _2},\;\;\;\;\;\;\;\;\;\;\;\;\;{\rm{           if }}{P_B}{\left| {{h_{B{R_m}}}} \right|^2} > {\Gamma _2}{\rm{ }}
\end{array} \right.,
\end{equation}
where ${\Gamma _2}$ is the saturated threshold of the harvester at $R_m$ .

\emph{\it 2) The second phase:} In this phase, $S$ respectively sends the signals ${x_{S{R_m}}}$ and ${x_{SE}}$ to ${R_m}$ and $E$ with ${\rm E}\left\{ {{{\left| {{x_{S{R_m}}}} \right|}^2}} \right\} = {\rm E}\left\{ {{{\left| {{x_{SE}}} \right|}^2}} \right\} = 1$. Considering IQI and CEEs, the received signals at ${R_m}$ and $E$ can be all written as (7) at the top of next page,
\begin{figure*}[!t]\label{7}
\normalsize
\begin{align}
{y_j} &= {\mu _{{r_j}}}\!\left[ {\left( {{{\hat h}_j}\! +\! {e_j}} \right)\!\left( {{\mu _{{t_j}}}\sqrt {{P_{S/R_m}}} {x_j} \!+\! {v_{{t_j}}}\sqrt {{P_{S/R_m}}} x_j^*} \right)\! + \! {n_j}}\right]\!+\! {v_{{r_j}}}{\left[ {\left( {{{\hat h}_j} \!+ \!{e_j}} \right)\left( {{\mu _{{t_j}}}\sqrt {{P_{S/R_m}}} {x_j} \!+\! {v_{{t_j}}}\sqrt {{P_{S/R_m}}} x_j^*}\right) \!+ \! {n_j}} \right]^*},
\end{align}
\hrulefill \vspace*{0pt}
\end{figure*}
\noindent{where ${\hat h_{S{R_m}}}$ and ${\hat h_{SE}}$ are the estimated channel coefficients from transmitter to receiver; ${n_{S{R_m}}} \sim \mathcal{CN}\left( {0,{N_{S{R_m}}}} \right)$ and ${n_{SE}} \sim \mathcal{CN}\left( {0,{N_{SE}}} \right)$ are the complex additive white Gaussian noise (AWGN). }

\emph{\it 3) The third phase: }In the third phase, ${R_m}$ respectively sends the signals ${x_{{R_m}D}}$, ${x_{{R_m}E}}$ to $D$ and $E$ with ${\rm E}\left\{ {{{\left| {{x_{{R_m}D}}} \right|}^2}} \right\} $ = ${\rm E}\left\{ {{{\left| {{x_{{R_m}E}}} \right|}^2}} \right\} = 1$. Similarly, the received signals at $D$ and $E$ can be expressed as (7) at the top of next page.{\footnote{Note that $j \in \left\{ {S{R_m},SE,{R_m}D,{R_m}E} \right\}$, and ${P_S}$ and ${P_{{R_m}}}$ are the power from $S$ and ${{R_m}}$, respectively.}}

Hence, the received signal-to-interference-plus-noise ratio (SINRs) at $R_m$, $D$ and $E$ can be expressed in a unified form as
\begin{equation}
\label{8}
{\gamma _j} = \frac{{{{\left| {{{\hat h}_j}} \right|}^2}{\rho _j}{p_j}}}{{\sigma _{{e_j}}^2{\rho _j}{p_j} + {{\left| {{{\hat h}_j}} \right|}^2}{\rho _j}{q_j} + \sigma _{{e_j}}^2{\rho _j}{q_j} + {g_j}}},
\end{equation}
where ${\rho _j} = {P_{S/R}}/{{{N_j}}}$, ${p_j} = {\left| {{\mu _{{t_j}}}{\mu _{{r_j}}} + v_{{t_j}}^*{v_{{r_j}}}} \right|^2}$, ${q_j} = {\left| {{\mu _{{r_j}}}{v_{{t_j}}} + \mu _{{t_j}}^*{v_{{r_j}}}} \right|^2}$ and ${g_j} = {\left| {{\mu _{{r_j}}} + {v_{{r_j}}}} \right|^2}$.

According to the Shannon's theorem, the channel capacity can be expressed as follows
\begin{equation}
\label{9}
{C_j} = \frac{{1 - \alpha }}{2}{\log _2}\left( {1 + {\gamma _j}} \right),
\end{equation}
where the factor $\frac{{1 - \alpha }}{2}$ means the data transmission is accomplished in equal two phases.

With DF protocol, the effective end-to-end capacity from $S$ to $D$ can be expressed as
\begin{equation}
\label{10}
{C_R} = \min \left( {{C_{S{R_m}}},{C_{{R_m}D}}} \right).
\end{equation}

\section{Performance Analysis}
This section analyzes the reliability and security of considered system in the presence of nonlinear energy harvester, IQI and ICSI. The closed-form expressions for the OP under the RRS, SRS, ORS schemes and IP under direct transmission and relay transmission strategies are derived.\footnote{The reliability and security are another metrics to characterize the PLS of wireless communication systems without using any secrecy coding, which are formulated by OP and IP \cite{7579030}.} Moreover, the asymptotic behaviors for the OP are analyzed, as well as the diversity orders.

\vspace{-4mm}
\subsection{Outage Probability Analysis}

In this subsection, the expressions for the OP are presented according the three RS strategies considered IQI, ICSI and nonlinear energy harvesters. The OP is defined as the probability that effective channel capacity is below the threshold ${R_{th}}$, which can be expressed as
\begin{equation}
\label{11}
{P_{out}} \buildrel \Delta \over = \Pr \left\{ {{C_R} < {R_{th}}} \right\}.
\end{equation}

\subsubsection{Random relay selection}

For RRS strategy, the link between $S$ and arbitrary one of the relay ${R_m}$ is selected, and the effective rate can be obtained as
\begin{equation}
\label{12}
{C_{{R_m}}} = \min \left( {{C_{S{R_m}}},{C_{{R_m}D}}} \right).
\end{equation}

Based on the above discussion, we can obtain the analytical expression for the OP of the RRS strategy in the following theorem.
\begin{theorem}\label{theorem:1} The analytical expression for the OP of RRS strategy is provided in (13) as shown at the top of next page.
\begin{figure*}[!t]\label{13}
\normalsize
\begin{align}\nonumber
\label{16}
&P_{out}^{RRS} \!=\! 1\! -\!\! \left(\! {\frac{{{\lambda \!_{S\!{R\!_m}}}}}{{{C\!_1}}}{e^{ -\! \frac{{{\lambda \!_{S\!{R\!_m}}}{C\!_2}}}{{{C\!_1}}}}}\!\left[\! {\sqrt {\!\frac{{{\beta \!_1}}}{{{\gamma \!_1}}}} {{\rm K}\!_1}\!\left(\! {\sqrt {{\beta \!_1}{\gamma \!_1}} } \!\right) \!-\! \frac{{\pi\! {\Lambda \!_1}}}{{2{Y\!_1}}}\!\sum\limits_{{l\!_1} \!= \!0}^{{Y\!_1}} {{e^{ - \!\frac{{2{\lambda \!_{B\!S}}{g\!_{S\!R\!_m}}{N\!_{S\!{R\!_m}}}\!\varepsilon }}{{{\Lambda \!_1}\!\left(\! {{\delta \!_{{l\!_1}}}\! +\! 1}\! \right)}}\! -\! \frac{{{\lambda \!_{S\!{R\!_m}}}{\Lambda \!_1}\!\left(\! {{\delta \!_{{l\!_1}}} \!+\! 1} \!\right)}}{{2{C\!_1}}}}}} \!\sqrt {1\! - \!\delta \!_{{l\!_1}}^2} }\!\right] \!+\! {e^{{\lambda \!_{B\!S}}{E\!_1}}}\!\left( {{e^{ -\!{\lambda \!_{S\!{R\!_m}}}{\Theta \!_1}}} \!- \!{e^{ -\! {\lambda \!_{S\!{R\!_m}}}{T\!_1}}}} \!\right)}\! \right)\\
&\;\;\; \times\! \left(\! \!{\frac{{{\lambda \!_{{R\!_m}\!D}}}}{{{C\!_3}}}{e^{ -\! \frac{{{\lambda \!_{{R\!_m}\!D}}{C\!_4}}}{{{C\!_3}}}}}\!\!\!\left[\! {\sqrt {\!\frac{{{\beta\! _2}}}{{{\gamma\! _2}}}} {{\rm K}\!_1}\!\sqrt {\!{\beta\! _2}{\gamma\! _2}}  \!- \!\frac{{\pi\! {\Lambda \!_2}}}{{2{Y\!_2}}}\!\sum\limits_{{l\!_2} \!=\! 0}^{{Y\!_2}} {{e^{ -\! \frac{{2{\lambda \!_{B\!{R\!_m}}}{g\!_{R\!_m\!D}}{N\!_{{R\!_m}\!D}}\!\varepsilon }}{{{\Lambda \!_2}\!\left( {{\delta \!_{{l\!_2}}} \!+\! 1} \!\right)}} \!-\! \frac{{{\lambda \!_{{R\!_m}\!D}}{\Lambda \!_2}\!\left(\! {{\delta \!_{{l\!_2}}}\! +\! 1} \!\right)}}{{2{C\!_3}}}}}} \!\!\sqrt {1 \!\!-\! \delta \!_{{l\!_2}}^2} } \!\right] \!\!+\! {e^{ -\! {\lambda\! _{B\!{R\!_m}}}\!{E\!_2}}}\!\!\left( {{e^{ -\! {\lambda \!_{{R\!_m}\!D}}\!{\Theta \!_2}}} \!\!-\!\! {e^{ -\! {\lambda \!_{{R\!_m}\!D}}{T\!_3}}}}\! \right)} \!\!\right),
\end{align}
\hrulefill \vspace*{0pt}
\end{figure*}
\end{theorem}
\noindent{where ${A_1} = \frac{{2\alpha {\varsigma _1}}}{{1 - \alpha }}$, ${E_1} = \frac{{{\Gamma _1}}}{{{P_B}}}$, ${C_1} = {A_1}{P_B}\left( {{p_{SR_m}} - {q_{SR_m}}\varepsilon } \right)$, ${C_2} = \sigma _{{e_{S{R_m}}}}^2{A_1}{P_B}\varepsilon \left( {{p_{SR_m}}{\rm{ + }}{q_{SR_m}}} \right)$, ${T_1} = \frac{{{g_{SR_m}}{N_{S{R_m}}}\varepsilon }}{{{C_1}{E_1}}} + \frac{{{C_2}}}{{{C_1}}}$, ${\beta _1} = 4{\lambda _{BS}}{g_{SR_m}}{N_{S{R_m}}}\varepsilon $, ${\gamma _1} = \frac{{{\lambda _{S{R_m}}}}}{{{C_1}}}$, ${\Lambda _1} = {C_1}{T_1} - {C_2}$, ${\delta _{{l_1}}} = \cos \left[ {\frac{{\left( {2{l_1} - 1} \right)\pi }}{{2{Y_1}}}} \right]$, ${\Theta _1} = \frac{{\varepsilon \sigma _{{e_{S{R_m}}}}^2{A_1}{\Gamma _1}\left( {{p_{SR_m}} + {q_{SR_m}}} \right) + \varepsilon {g_{SR_m}}{N_{S{R_m}}}}}{{{A_1}{\Gamma _1}\left( {{p_{SR_m}} - \varepsilon {q_{SR_m}}} \right)}}$, ${A_2} = \frac{{2\alpha {\varsigma _2}}}{{1 - \alpha }}$, ${E_2} = \frac{{{\Gamma _2}}}{{{P_B}}}$, ${C_3} = {A_2}{P_B}\left( {{p_{R_mD}} - {q_{R_mD}}\varepsilon } \right)$, ${C_4} = \sigma _{{e_{{R_m}D}}}^2{A_2}{P_B}\varepsilon \left( {{p_{R_mD}}{\rm{ + }}{q_{R_mD}}} \right)$, ${T_3} = \frac{{{g_{R_mD}}{N_{{R_m}D}}\varepsilon }}{{{C_3}{E_2}}} + \frac{{{C_4}}}{{{C_3}}}$, ${\beta _2} = 4{\lambda _{B{R_m}}}{g_{R_mD}}{N_{{R_m}D}}\varepsilon $, ${\gamma _2} = \frac{{{\lambda _{{R_m}D}}}}{{{C_3}}}$, ${\Lambda _2} = {C_3}{T_3} - {C_4}$, ${\delta _{{l_2}}} = \cos \left[ {\frac{{\left( {2{l_2} - 1} \right)\pi }}{{2{Y_2}}}} \right]$ and ${\Theta _2} = \frac{{\varepsilon \sigma _{{e_{{R_m}D}}}^2{A_2}{\Gamma _2}\left( {{p_{R_mD}} + {q_{R_mD}}} \right) + \varepsilon {g_{R_mD}}{N_{{R_m}D}}}}{{{A_2}{\Gamma _2}\left( {{p_{R_mD}} - \varepsilon {q_{R_mD}}} \right)}}$.}

For $\varepsilon  < {1 \mathord{\left/
 {\vphantom {1 {\max \left\{ {\frac{{{p_{S{R_m}}}}}{{{q_{S{R_m}}}}},\frac{{{p_{{R_m}D}}}}{{{q_{{R_m}D}}}}} \right\}}}} \right.
 \kern-\nulldelimiterspace} {\max \left\{ {\frac{{{p_{S{R_m}}}}}{{{q_{S{R_m}}}}},\frac{{{p_{{R_m}D}}}}{{{q_{{R_m}D}}}}} \right\}}}$, otherwise the OP expressions are equal to 1.
\begin{proof}
See Appendix A.
\end{proof}

To get deeper insights, the asymptotic behavior of non-ideal conditions $(\sigma _{{e_{S{R_m}}}}^2 = \sigma _{{e_{{R_m}D}}}^2 = t)$ is investigated at high SNRs in the following corollary.

\begin{corollary}
The asymptotic expression of OP for RRS strategy under non-ideal conditions $(\sigma _{{e_{S{R_m}}}}^2 = \sigma _{{e_{{R_m}D}}}^2 = t)$ is given by
\begin{equation}\label{14}
P_{out}^{RRS,\infty} = 1 - {e^{ - {\lambda _{S{R_m}}}{H_1} - {\lambda _{{R_m}D}}{H_2}}},
\end{equation}
{where ${H_1} = {{\varepsilon \sigma _{{e_{S{R_m}}}}^2\left( {{p_{SR_m}} + {q_{SR_m}}} \right)}}/{{{p_{SR_m}} - \varepsilon {q_{SR_m}}}}$ and ${H_2} = {{\varepsilon \sigma _{{e_{{R_m}D}}}^2\left( {{p_{R_mD}} + {q_{R_mD}}} \right)}}/{{({p_{R_mD}} - \varepsilon {q_{R_mD}}})}$.}
\end{corollary}
\begin{proof}
Based on (9), the asymptotic channel capacities of $S \to {R_m}$ and ${R_m} \to D$ can be written as
\begin{equation}\label{15}
C_{S\!{R\!_m}}^{\infty ,ni} \!\!=\!\! \frac{{1 \!\!-\!\! \alpha }}{2}{\log \!_2}\!\!\left(\!\! {1\!\! +\!\! \frac{{{{\left| {{{\hat h}\!_{S\!{R\!_m}}}} \!\right|}^2}{p\!_{S\!R\!_m}}}}{{\sigma \!_{{e\!_{S\!{R\!_m}}}}^2\!{p\!_{S\!R\!_m}} \!\!+\!\! {{\left| {{{\hat h}\!_{S\!{R\!_m}}}} \right|}^2}\!\!{q\!_{S\!R\!_m}} \!\!+\!\! \sigma\! _{{e\!_{S\!{R\!_m}}}}^2\!\!{q\!_{S\!R\!_m}}}}} \!\!\right),
\end{equation}
\begin{equation}\label{16}
C_{{R\!_m}\!D}^{\infty ,ni} \!\!=\!\! \frac{\!{1\!\! -\!\! \alpha }}{2}{\log \!_2}\!\!\left(\!\! {1\!\! +\!\! \frac{{{{\left| {{{\hat h}\!_{{R\!_m}\!D}}} \right|}^2}\!{p\!_{R\!_m\!D}}}}{{\sigma \!_{{e\!_{{R\!_m}\!D}}}^2\!{p\!_{R\!_m\!D}}\!\! +\! \!{{\left| {{{\hat h}\!_{{R\!_m}\!D}}} \right|}^2}\!{q\!_{R\!_m\!D}} \!\!+\!\! \sigma \!_{{e\!_{{R\!_m}\!D}}}^2\!{q\!_{R\!_m\!D}}}}} \!\!\right).
\end{equation}

By the definition of OP, the following expression can be obtained as
\begin{align}\label{17}\nonumber
P_{out}^{\infty ,ni} &= \Pr \left\{ {\min \left( {C_{S{R_m}}^{\infty ,ni},C_{{R_m}D}^{\infty ,ni}} \right) < {R_{th}}} \right\}\\
& = 1\! -\! \Pr \!\left\{ {C\!_{S\!{R\!_m}}^{\infty ,ni} \!>\! {R\!_{th}}} \!\right\}\!\Pr\! \left\{ {C\!_{{R\!_m}\!D}^{\infty ,ni} \!>\! {R\!_{th}}} \!\right\}.
\end{align}

Utilizing the similar methodology of Appendix A, (14) can be derived.
\end{proof}

Furthermore, the diversity order is investigated, which can be defined as \cite{8879698}:
\begin{equation}\label{18}
d =  - \mathop {\lim }\limits_{{\rho _j} \to \infty } \frac{{\log \left( {P_{out}^\infty } \right)}}{{\log {\rho _j}}},
\end{equation}
where $\rho _j$ is the average SNR and ${P_{out}^\infty }$ is the asymptotic OP.

\begin{corollary}
The diversity order of OP for RRS scheme in the presence of non-ideal conditions ($\sigma _{{e_{S{R_m}}}}^2 = \sigma _{{e_{{R_m}D}}}^2 = t$) can be obtained as follows:
\begin{equation}\label{19}
d_{RRS}^{ni}\left( {{\rho _{S{R_m}}},{\rho _{{R_m}D}}} \right) = 0.
\end{equation}
\end{corollary}
\begin{proof}
Follows trivially by using (18) and the definition of derivative.
\end{proof}

\begin{remark}\label{remark 1}
From \textbf{Theorem 1}, \textbf{Corollary 1} and \textbf{Corollary 2}, the following observations can be obtained as: 1) When $M$ increases gradually, it can be seen that (13) and (14) are independent of $M$, so the RRS scheme will not change with the increase or decrease of the number of antennas; 2) At high average SNR, $P_{out}^{RRS,\infty}$ is a fixed constant, which results in 0 diversity order. This means that the diversity order can not be improved by increasing the number of relays.
\end{remark}

\subsubsection{Suboptimal relay selection}

For SRS strategy, the optimal relay is selected according to maximizing the capacity of the link $S \to R_m $, which can be expressed as
\begin{equation}\label{18}
a = \arg \mathop {\max }\limits_{m = 1,2, \cdots M} {C_{S{R_m}}},
\end{equation}
\begin{equation}\label{19}
{C_{{R_a}}} = \min \left( {{C_{S{R_a}}},{C_{{R_a}D}}} \right).
\end{equation}

Based on (11) and (20), we have the following Theorem 2.
\begin{theorem}\label{theorem:2}The analytical expression of OP is provided for SRS strategy in (22) as shown at the top of next page.
\begin{figure*}[!t]\label{20}
\normalsize
\begin{align}\nonumber
\label{23}\nonumber
P_{out}^{SRS} &= 1 + \left( {\frac{\Xi }{{{C_5}}}{e^{ - \frac{{{\lambda _{S{R_a}}}\left( {s + 1} \right){C_6}}}{{{C_5}}}}}\left[ {\sqrt {\frac{{{\beta _3}}}{{{\gamma _3}}}} {{\rm K}_1}\left( {\sqrt {{\beta _3}{\gamma _3}} } \right) - \frac{{\pi {\Lambda _3}}}{{2{Y_3}}}\sum\limits_{{l_3} = 0}^{{Y_3}} {{e^{ - \frac{{2{\lambda _{BS}}{g_{SR_m}}{N_{S{R_a}}}\varepsilon }}{{{\Lambda _3}\left( {{\delta _{{l_3}}} + 1} \right)}} - \frac{{{\lambda _{S{R_a}}}\left( {s + 1} \right){\Lambda _3}\left( {{\delta _{{l_3}}} + 1} \right)}}{{2{C_5}}}}}\sqrt {1 - \delta _{{l_3}}^2} } } \right]} \right.\\\nonumber
&\;\;\;\;\left. { -\! \frac{\Xi }{{{\lambda\! _{S\!{R\!_a}}}\!\left(\! {s\! +\! 1} \!\right)}}{e^{ - \!{\lambda\! _{B\!S}}{E\!_1} \!-\! {\lambda \!_{S\!{R\!_a}}}\!\left(\! {s\! +\! 1} \!\right){T\!_5}}} \!- \!\left[\!{1 \!-\! {{\left(\! {1\!-\! {e^{ -\! {\lambda \!_{S\!{R\!_a}}}{\Theta\! _3}}}}\! \right)}^M}} \!\right]{e^{ -\! {\lambda \!_{B\!S}}{E\!_1}}}} \!\right) \!\times \left(\! {{e^{ -\! {\lambda \!_{B\!{R\!_a}}}{E\!_2}}}\left( {{e^{ - {\lambda _{{R_a}D}}{\Theta _4}}} - {e^{ - {\lambda _{{R_a}D}}{T_7}}}} \right) + } \right.\\
&\;\;\;\;\left. {\frac{{{\lambda _{{R_a}D}}}}{{{C_7}}}{e^{ - \frac{{{\lambda _{{R_a}D}}{C_8}}}{{{C_7}}}}}\left[ {\sqrt {\frac{{{\beta _4}}}{{{\gamma _4}}}} {{\rm K}_1}\sqrt {{\beta _4}{\gamma _4}}  - \frac{{\pi {\Lambda _4}}}{{2{Y_4}}}\sum\limits_{{l_4} = 0}^{{Y_4}} {{e^{ - \frac{{2{\lambda _{B{R_a}}}{g_{R_mD}}{N_{{R_a}D}}\varepsilon }}{{{\Lambda _4}\left( {{\delta _{{l_4}}} + 1} \right)}} - \frac{{{\lambda _{{R_a}D}}{\Lambda _4}\left( {{\delta _{{l_4}}} + 1} \right)}}{{2{C_7}}}}}} \sqrt {1 - \delta _{{l_4}}^2} } \right]} \right),
\end{align}
\hrulefill \vspace*{0pt}
\end{figure*}

\noindent{where $\Xi  =  - M{\lambda _{S{R_a}}}\sum\limits_{s = 0}^{M - 1} {\left( {\begin{array}{*{20}{c}}
{M - 1}\\
s
\end{array}} \right){{\left( { - 1} \right)}^s}} $, ${C\!_5} \!=\! {A\!_1}{P\!_B}\left( {{p_{SR_a}} - {q_{SR_a}}\varepsilon } \right)$, ${C_6} = \sigma _{{e_{S{R_a}}}}^2{A_1}{P_B}\varepsilon \left( {{p_{SR_a}} + {q_{SR_a}}} \right)$, ${T_5} = \frac{{{g_{SR_a}}{N_{S{R_a}}}\varepsilon }}{{{C_5}{E_1}}} + \frac{{{C_6}}}{{{C_5}}}$, ${\beta _3} = 4{\lambda _{BS}}{g_{SR_a}}{N_{S{R_a}}}\varepsilon $, ${\gamma _3} = \frac{{{\lambda _{S{R_a}}}\left( {s + 1} \right)}}{{{C_5}}}$, ${\Lambda _3} = {C_5}{T_5} - {C_6}$, ${\delta _{{l_3}}} = \cos \left[ {\frac{{\left( {2{l_3} - 1} \right)\pi }}{{2{Y_3}}}} \right]$, ${C_7} = {A_2}{P_B}\left( {{p_{R_aD}} - {q_{R_aD}}\varepsilon } \right)$, ${C_8} = \sigma _{{e_{{R_a}D}}}^2{A_2}{P_B}\varepsilon \left( {{p_{R_aD}} + {q_{R_aD}}} \right)$, ${T_7} = \frac{{{g_{R_aD}}{N_{{R_a}D}}\varepsilon }}{{{C_7}{E_2}}} + \frac{{{C_8}}}{{{C_7}}}$, ${\beta _4} = 4{\lambda _{B{R_a}}}{g_{R_aD}}{N_{{R_a}D}}\varepsilon $, ${\gamma _4} = \frac{{{\lambda _{{R_a}D}}}}{{{C_7}}}$, ${\Lambda _4} = {C_7}{T_7} - {C_8}$ and ${\delta _{{l_4}}} = \cos \left[ {\frac{{\left( {2{l_4} - 1} \right)\pi }}{{2{Y_4}}}} \right]$.}
\end{theorem}
\begin{proof}
See Appendix B.
\end{proof}

Similarly, the asymptotic behavior of non-ideal conditions is studied of OP for SRS strategy in the high SNR regime.
\begin{corollary}
The asymptotic expression for the OP of SRS strategy under non-ideal conditions $(\sigma _{{e_{S{R_a}}}}^2 = \sigma _{{e_{{R_a}D}}}^2 = t)$ is given by
\begin{equation}\label{21}
P_{out}^{SRS,\infty } = 1 - \left( {1 - {{\left( {1 - {e^{ - {\lambda _{S{R_a}}}{H_3}}}} \right)}^M}} \right){e^{ - {\lambda _{{R_a}D}}{H_4}}},
\end{equation}
{where ${H_3} = {{\varepsilon \sigma _{{e_{S{R_a}}}}^2\left( {{p_{SR_a}} + {q_{SR_a}}} \right)}}/{{{p_{SR_a}} - \varepsilon {q_{SR_a}}}}$ and ${H_4} = {{\varepsilon \sigma _{{e_{{R_a}D}}}^2\left( {{p_{R_aD}} + {q_{R_aD}}} \right)}}/{({{p_{R_aD}} - \varepsilon {q_{R_aD}}})}$.}
\end{corollary}

Then, the diversity order of OP for SRS strategy under non-ideal conditions ($\sigma _{{e_{S{R_a}}}}^2 = \sigma _{{e_{{R_a}D}}}^2 = t$) is presented in the following corollary.
\begin{corollary}
The diversity order of OP for SRS scheme in the presence of non-ideal conditions ($\sigma _{{e_{S{R_m}}}}^2 = \sigma _{{e_{{R_m}D}}}^2 = t$) is given by:
\begin{equation}\label{24}
d_{SRS}^{ni}\left( {{\rho _{S{R_a}}},{\rho _{{R_a}D}}} \right) = 0.
\end{equation}
\end{corollary}

\begin{remark}\label{remark 2}
From \textbf{Theorem 2}, \textbf{Corollary 3} and \textbf{Corollary 4}, we can obtain the following information as: 1) when the number of relay increases, it can be concluded from formulas (22) and (23) that the system's outage performance becomes better under the SRS strategy; 2) From expression (22), it can be obtained that when the $M$ is fixed and the transmit power at $B$ is in a high state, the OP will cause an error floor; 3) From (24), we can observe that the diversity order of the considered system is zero due to the fixed constant for the OP in the high SNR regime.
\end{remark}

\subsubsection{Optimal relay selection}
For ORS strategy, the optimal relay is selected according to maximize the capacity of the links both $S \to R_m$ and $R_m \to D$
\begin{equation}\label{22}
{m^ * } = \arg \mathop {\max }\limits_{1 \le m \le M} \min \left\{ {{C_{S{R_m}}},{C_{{R_m}D}}} \right\},
\end{equation}
\begin{equation}\label{23}
{C_{{R_{{m^ * }}}}} = \mathop {\max }\limits_{1 \le m \le M} {C_{{R_m}}}.
\end{equation}

According to (11) and (26), Theorem 3 can be obtained as following.
\begin{theorem}\label{theorem:3}The analytical expression of the OP is provided for the ORS strategy in (27) as shown at the top of next page.
\begin{figure*}[!t]\label{24}
\normalsize
\begin{align}\nonumber
\label{24}\nonumber
&P_{out}^{O\!R\!S}\! \!=\!\! \prod\limits_{m\! = \!1}^M {\!\left\{ {1\!\! -\! \!\left(\! {\frac{{{\lambda \!_{S\!{R\!_m}}}}}{{{C\!_1}}}{e^{ - \!\frac{{{\lambda \!_{S\!{R\!_m}}}{C\!_2}}}{{{C\!_1}}}}}\!\!\left[\! {\sqrt {\!\frac{{{\beta \!_1}}}{{{\gamma\! _1}}}} {{\rm K}\!_1}\!\!\left(\! {\sqrt {\!{\beta \!_1}{\gamma \!_1}} } \!\!\right) \!\!- \!\!\frac{{\pi\! {\Lambda \!_1}}}{{2{Y\!_1}}}\!\!\sum\limits_{{l\!_1} \!=\! 0}^{{Y\!_1}} {{e^{ -\! \frac{{2{\lambda \!_{B\!S}}{g\!_{S\!R\!_m}}{N\!_{S\!{R\!_m}}}\!\varepsilon }}{{{\Lambda \!_1}\!\left( \!{{\delta\! _{{l\!_1}}} \!+\! 1} \!\right)}} \!- \frac{{{\lambda \!_{S\!{R\!_m}}}{\Lambda\! _1}\!\left(\! {{\delta\! _{{l\!_1}}} \!+\! 1} \!\right)}}{{2{C\!_1}}}}}}\!\!\! \sqrt {1\! -\! \delta _{{l\!_1}}^2} } \!\right]\! +\! {e^{{\lambda \!_{B\!S}}{E\!_1}}}\left(\! {{e^{ -\! {\lambda \!_{S\!{R\!_m}}}\!{\Theta \!_1}}}\! -\! {e^{ -\! {\lambda \!_{S\!{R\!_m}}}\!{T\!_1}}}}\! \right)} \!\right)}\! \right.} \\
&\;\;\;\;\left. { \times\!\! \left(\! {\frac{{{\lambda \!_{{R\!_m}\!D}}}}{{{C\!_3}}}{e^{ -\! \frac{{{\lambda \!_{{R\!_m}\!D}}\!{C\!_4}}}{{{C\!_3}}}}}\!\!\left[\! {\sqrt {\!\!\frac{{{\beta \!_2}}}{{{\gamma \!_2}}}} {{\rm K}\!_1}\!\sqrt {{\beta \!_2}{\gamma \!_2}}  \!- \!\frac{{\pi\! {\Lambda \!_2}}}{{2{Y\!_2}}}\!\!\sum\limits_{{l\!_2} \!=\! 0}^{{Y\!_2}} {{e^{ -\! \frac{{2{\lambda \!_{B\!{R\!_m}}}\!{g\!_{R\!_m\!D}}{N\!_{{R\!_m}\!D}}\!\varepsilon }}{{{\Lambda \!_2}\!\left(\! {{\delta \!_{{l\!_2}}} \!+\! 1} \!\right)}} \!\!-\!\! \frac{{{\lambda\! _{{R\!_m}\!D}}\!{\Lambda \!_2}\!\left(\! {{\delta _{{l\!_2}}} \!+\! 1} \!\right)}}{{2{C\!_3}}}}}}\!\! \sqrt {\!1\! -\! \delta _{{l\!_2}}^2} } \!\right] \!\!+\!\! {e^{ -\! {\lambda \!_{B\!{R\!_m}}}\!{E\!_2}}}\!\!\left( \!{{e^{ -\! {\lambda \!_{{R\!_m}\!D}}\!{\Theta \!_2}}} \!-\! {e^{ -\! {\lambda \!_{{R\!_m}\!D}}\!{T\!_3}}}} \!\right)} \!\!\right)}\!\! \right\},
\end{align}
\hrulefill \vspace*{0pt}
\end{figure*}

\end{theorem}
\begin{proof}
See Appendix C.
\end{proof}

Next, the asymptotic behavior for the OP of ORS strategy in the presence of non-ideal conditions is studied.
\begin{corollary}
The asymptotic expression of OP for the ORS strategy under non-ideal conditions $(\sigma _{{e_{S{R_a}}}}^2 = \sigma _{{e_{{R_a}D}}}^2 = t)$ is given by
\begin{equation}\label{25}
P_{out}^{ORS,\infty} = \prod\limits_{i = 1}^M {\left( {1 - {e^{ - {\lambda _{S{R_m}}}{H_1} - {\lambda _{{R_m}D}}{H_2}}}} \right)} .
\end{equation}
\end{corollary}

\begin{corollary}
The diversity order of OP for ORS strategy under non-ideal conditions ($\sigma _{{e_{S{R_m}}}}^2 = \sigma _{{e_{{R_m}D}}}^2 = t$) is following:
\begin{equation}\label{29}
d_{ORS}^{ni}\left( {{\rho _{S{R_m}}},{\rho _{{R_m}D}}} \right) = 0.
\end{equation}
\end{corollary}

\begin{remark}\label{remark 3}
From \textbf{Theorem 3}, \textbf{Corollary 5}  and \textbf{Corollary 6}, we can get the following points as: 1) When $M$ increases, $P_{out}^{ORS}$  and $P_{out}^{ORS,\infty }$ will become smaller because, which means that the system's outage performance becomes better under the ORS strategy; 2) As $P_B$ goes to infinity, the OP of the considered system under non-ideal conditions has an error floor; 3) We can also observe that the diversity order is 0, which means that the slope of the outage probability is 0.
\end{remark}

\vspace{-4mm}
\subsection{Intercept Probability Analysis}

In this subsection, the secrecy performance of the multi-relay networks with IQI is studied in terms of IP considering two scenarios of direct transmission and transmission via relay. The definition of IP is the probability that the channel capacity between $S \to E$ or ${R_m} \to E$ is greater than the threshold ${R_{th}}$, which can be formulated as
\begin{equation}\label{26}
P_{{\rm{int}}}^{{\rm{direct/relay}}} \buildrel \Delta \over = \Pr \left\{ {{C_{SE/{R_c}E}} > {R_{th}}} \right\},
\end{equation}
where ${R_c}$ is the selected relay, ${C_{SE}}$ and ${C_{{R_c}E}}$ are the intercept capacities of $S \to E$ and ${R_c} \to E$, respectively.

\subsubsection{Direct Transmission}
Under the condition of direct transmission, based on (8) and the definition of (30), the closed-form analytical expression of IP under the condition of direct transmission can be obtained as Theorem 4.

\begin{theorem}\label{theorem 4}
The analytical expression of IP under the condition of direct transmission is provided in (31) as shown at the top of next page.
\begin{figure*}[!t]\label{31}
\normalsize
\begin{align}
\label{27}
P_{{\rm{int}}}^{{\rm{direct}}} \!\!=\!\!  -\! {e^{ -\! {\lambda \!_{B\!S}}\!{E\!_1} \!-\! {\lambda \!_{S\!E}}\!{T\!_9}}} \!\!+ \!\!\frac{{{\lambda \!_{S\!E}}}}{{{C\!_9}}}{e^{ -\! \frac{{{\lambda \!_{S\!E}}\!{C\!_{10}}}}{{{C\!_9}}}}}\!\!\left[\! {2\!\sqrt {\!\frac{{{\beta \!_5}}}{{{\gamma \!_5}}}} {{\rm K}\!_1}\!\left(\! {2\!\sqrt {\!{\beta \!_5}{\gamma \!_5}} } \!\right) \!\!-\!\! \frac{{\pi\! {\Lambda \!_5}}}{{2{Y\!_5}}}\!\!\sum\limits_{{l\!_5} \!=\! 0}^{{Y\!_5}} \!{{e^{ -\! \frac{{{\gamma \!_5}\!{\Lambda \!_5}\!\left( \!{{\delta \!_{{l\!_5}}} \!+ \!1\!} \!\right)}}{2} \!-\! \frac{{2{\beta\! _5}}}{{{\Lambda \!_5}\!\left( {{\delta\! _{{l\!_5}}}\!\! +\!\! 1}\! \right)}}}}\!\sqrt {1\! -\! \delta\! _{{l\!_5}}^2} } } \!\right] \!+\! {e^{ -\! {\lambda \!_{S\!E}}{\Theta\! _5}\! -\! {\lambda \!_{B\!S}}{E\!_1}}},
\end{align}
\hrulefill \vspace*{0pt}
\end{figure*}

\noindent{where ${C_9} = {A_1}{P_B}\left( {{p_{SE}} - {q_{SE}}\varepsilon } \right)$, ${C_{10}} = \sigma _{{e_{SE}}}^2{A_1}{P_B}\varepsilon \left( {{p_{SE}} + {q_{SE}}} \right)$, ${T_9} = {{\varepsilon {g_{SE}}{N_{SE}}}}/{{{C_9}{E_1}}} + {{{C_{10}}}}/{{{C_9}}}$, ${\beta _5} = {\lambda _{BS}}{g_{SE}}{N_{SE}}\varepsilon $, ${\gamma _5} = \frac{{{\lambda _{SE}}}}{{{C_9}}}$, ${\Lambda _5} = {C_9}{T_9} - {C_{10}}$, ${\delta _{{l_5}}} = \cos \left[ {{{\left( {2{l_5} - 1} \right)\pi }}/{{2{Y_5}}}} \right]$ and ${\Theta _5} = {{\varepsilon \sigma _{{e_{SE}}}^2{A_1}{\Gamma _1}\left( {{p_{SE}} + {q_{SE}}} \right) + \varepsilon {g_{SE}}{N_{SE}}}}/{{{A_1}{\Gamma _1}\left( {{p_{SE}} - \varepsilon {q_{SE}}} \right)}}$.}
\end{theorem}
\begin{proof}
See Appendix D.
\end{proof}

\subsubsection{Transmission via Relay}
We then studied the security of the considered system by utilizing relay to transmit information in the following theorem.

\begin{theorem}\label{theorem 5}
The analytical expression of IP under the transmission via relay condition is provided in (32) as shown at the top of next page.
\begin{figure*}[!t]\label{32}
\normalsize
\begin{align}
\label{28}
P_{{\mathop{\rm int}} }^{{\rm{relay}}} \!\!=\!\! {e^{ -\! {\lambda \!_{B{R\!_c}}}\!{E\!_2}}}\!\!\left( \!{{e^{ - {\lambda\! _{{R\!_c}\!E}}{\Theta \!_6}}} \!-\! {e^{ -\! {\lambda\! _{{R\!_c}\!E}}\!{T\!_{11}}}}} \!\right) \!\!+\!\! \frac{{{\lambda\! _{{R\!_c}\!E}}}}{{{C\!_{11}}}}{e^{ -\! \frac{{{\lambda \!_{{R\!_c}\!E}}\!{C\!_{12}}}}{{{C\!_{11}}}}}}\!\!\!\left[\! {\sqrt {\!\frac{{{\beta \!_6}}}{{{\gamma \!_6}}}} {{\rm K}_1}\!\sqrt {\!{\beta\! _6}{\gamma \!_6}}  \!- \!\frac{{\pi\! {\Lambda \!_6}}}{{2{Y\!_6}}}\!\!\sum\limits_{{l\!_6}\! =\! 0}^{{Y\!_6}}\! {{e^{ - \!\frac{{{\beta \!_6}}}{{2{\Lambda \!_6}\!\left(\! {{\delta \!_{{l\!_6}}} \!+\! 1} \!\right)}} \!-\! \frac{{{\gamma \!_6}\!{\Lambda \!_6}\!\left( \!{{\delta _{{l\!_6}}}\! +\! 1} \!\right)}}{2}}}} \!\sqrt {1 \!-\! \delta _{{l\!_6}}^2} } \!\right],
\end{align}
\hrulefill \vspace*{0pt}
\end{figure*}
\end{theorem}

\begin{figure}[!t]
\centering
\includegraphics [width=3.2in]{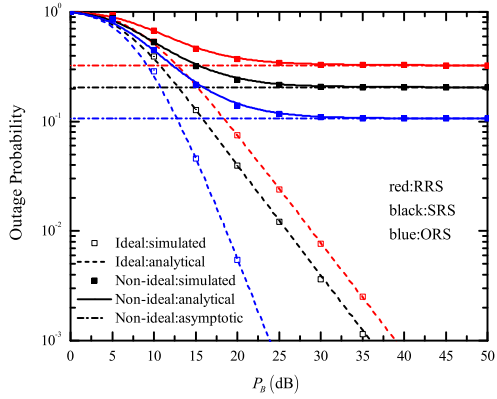}
\caption{{OP versus the transmit power for different relay selection strategies.}}
\label{Figure 2}
\end{figure}

\noindent{where ${C_{11}} = {A_2}{P_B}\left( {{p_{R_cE}} - {q_{R_cE}}\varepsilon } \right)$, ${C_{12}} = \sigma _{{e_{{R_c}E}}}^2{A_2}{P_B}\varepsilon \left( {{p_{R_cE}} + {q_{R_cE}}} \right)$, ${T_{11}} = {{{g_{R_cE}}{N_{{R_c}E}}\varepsilon }}/{{{C_{11}}{E_2}}} + {{{C_{12}}}}/{{{C_{11}}}}$, ${\beta _6} = 4{\lambda _{B{R_c}}}{g_{R_cE}}{N_{{R_c}E}}\varepsilon $, ${\gamma _6} = {{{\lambda _{{R_c}E}}}}/{{{C_{11}}}}$, ${\Lambda _6} = {C_{11}}{T_{11}} - {C_{12}}$, ${\delta _{{l_6}}} = \cos \left[ {{{\left( {2{l_6} - 1} \right)\pi }}/{{2{Y_6}}}} \right]$ and ${\Theta _6} = {{\varepsilon \sigma _{{e_{{R_c}E}}}^2{A_2}{\Gamma _2}\!\left( {{p_{R_cE}}\! + \!{q_{R_cE}}} \right) \!+\! \varepsilon {g_{R_cE}}{N_{{R_c}E}}}}/\!{{{A_2}{\Gamma _2}\left( {{p_{R_cE}} \!-\! \varepsilon {q_{R_cE}}} \right)}}$}.

\begin{proof}
See Appendix E.
\end{proof}

\section{Numerical Results}

In this section, some numerical results are provided to validate the correctness of the obtained results in the above section. The results are then verified using Monte Carlo simulations with $10^7$ iterations. Unless otherwise specified, we set the parameters as in Table I. 
\begin{table*}[!htb]
  \begin{center}
  \caption{Parameters for numerical results}
    \begin{tabular}{|l|l|}
      \hline
      Monte Carlo simulations repeated & ${10^7}$ iterations \\
      \hline
      Distance between nodes& ${d_{{R_m}D}} = {d_{{R_m}E}} = 1.5$, $d_{SE}=2$ \\
      \hline
      Shadow fading parameter&  $\beta  = 3$ \\
      \hline
      Time allocation factor&  $\alpha  = 0.5$ \\
      \hline
      Noise power  & ${N_{S{R_m}}}=N_{SE} = {N_{{R_m}D}}{\rm{ = }}{N_{{R_m}E}} = 1$ \\
      \hline
      Intercept capacity threshold  & ${R_{th}} = 0.05$ \\
      \hline
      \hline
      Amplitude at TX and RX  & ${\xi _t} = {\xi _r} = \{1, 1.1\}$, \\
      \hline
      \hline
      Phase at TX and RX  & ${\phi _t} = {\phi _r} = \{0^ \circ, 5^ \circ \}$  \\
      \hline
      Variance of CEEs  & ${\varsigma _1} = {\varsigma _2} = 0.5$  \\
      \hline
      Energy converse coefficient at source and relay  & ${\sigma_e^{2}} = \{0, 0.05 \}$  \\
      \hline

    \end{tabular}

  \end{center}
\end{table*}

\begin{figure}[!t]
\centering
\subfigure[ ]
{\begin{minipage}[t]{0.45\linewidth}
\centering
\includegraphics[width= 2in]{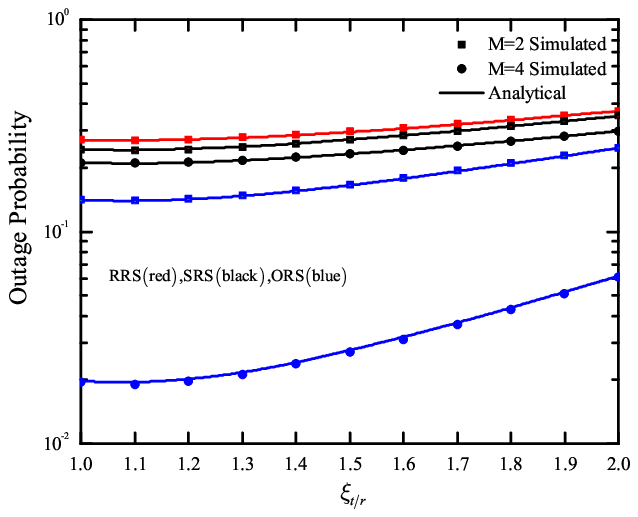}
\end{minipage}
}
\centering
\subfigure[ ]
{\begin{minipage}[t]{0.48\linewidth}
\centering
\includegraphics[width= 2in]{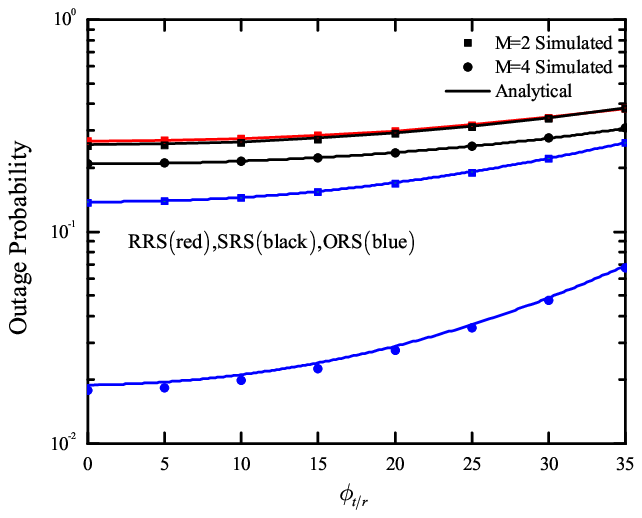}
\end{minipage}
}
\caption{Influence of IQI: (a) OP versus TX/RX amplitude; (b) OP versus phase mismatch.}
\end{figure}

\begin{figure}[!t]
\centering
\includegraphics [width=3.2in]{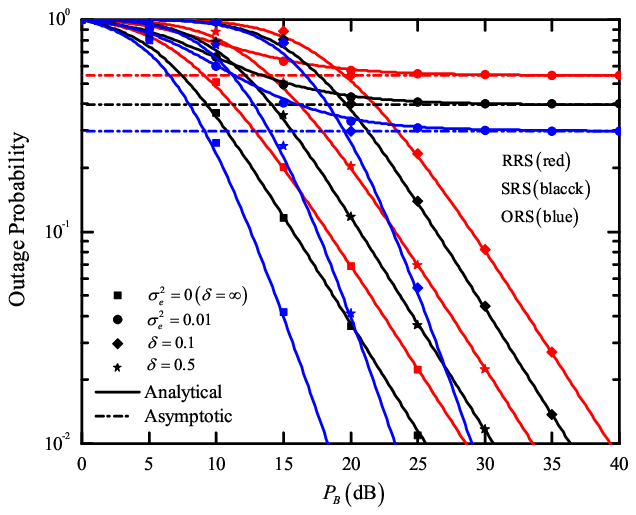}
\caption{{OP versus the transmit power for different CEE parameters.}}
\label{Figure 4}
\end{figure}

\begin{figure}[!t]
\centering
\includegraphics [width=3.2in]{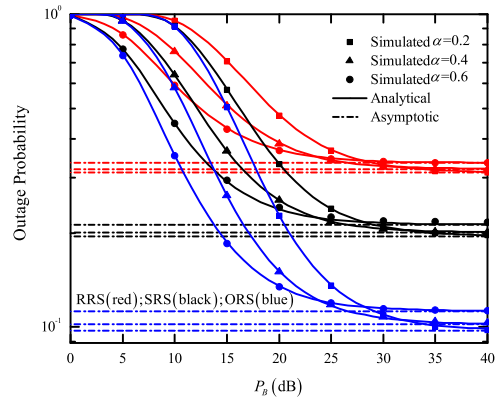}
\caption{{OP versus the transmit power for different time allocation factors.}}
\label{Figure 5}
\end{figure}
\subsection{Reliability Analysis}
Fig. 2 plots the OP versus the transmit power ${P_B}$ for different RS strategies. For the purpose of comparison, the curves of ideal conditions are provided. We set $M = 2$. These simulation results perfectly verify the derived closed-form analytical expressions of (13), (22) and (27) and asymptotic expressions of (14), (23) and (28), as well as (19), (24) and (29). We can also see from the simulation results that: 1) OP under the non-ideal conditions is greater than that the ideal conditions due to the IQI and CEEs; 2) The outage performance under RRS strategy is worse than SRS and ORS strategies, and ORS scheme has the best outage performance; 3) There are error floors of the OP for the three RS schemes at high SNRs due to CEEs, which means that the system OP performance can not always be improved by increasing transmit power.

Fig. 3 (a) illustrates the OP of TX/RX amplitude ${\xi _{{t \mathord{\left/ {\vphantom {t r}} \right. \kern-\nulldelimiterspace} r}}}$ for different number of relays ($M = \{2,4\}$) under two RS strategies. We set ${P_B} = 20$ dB. These results indicate that the OP for SRS scheme is higher than that of ORS scheme for arbitrary number of relay ($M > 1$). The gap of OP between the two schemes becomes large as the number of relays increases. Also, we can see that the outage performance of the considered system is proportional to the $M$. Finally, the OP of the system increases gradually with the increase of TX/RX amplitude, which means that the ${\xi _{{t \mathord{\left/ {\vphantom {t r}} \right. \kern-\nulldelimiterspace} r}}}$ has negative effects on the system performance.
Fig. 3 (b) plots the OP versus phase mismatch ${\phi _{{t \mathord{\left/ {\vphantom {t r}} \right. \kern-\nulldelimiterspace} r}}}$ for different number of relays under two RS strategies. As in Fig. 3 (a), we set ${P_B} = 20$dB. These simulation results verify that with the increase of ${\phi _{{t \mathord{\left/ {\vphantom {t r}} \right.  \kern-\nulldelimiterspace} r}}}$, the outage performance of the system gradually becomes worse. Furthermore, the effects of the number of relays on the system performance in Fig. 2 and Fig. 3  are further verified. From Fig. 3 (a) and Fig. 3 (b), we can observe that the parameters of amplitude and phase mismatches have the same effects on the outage performance.

\begin{figure}[!t]
\centering
\subfigure[ ]
{\begin{minipage}[t]{0.45\linewidth}
\centering
\includegraphics[width= 2in]{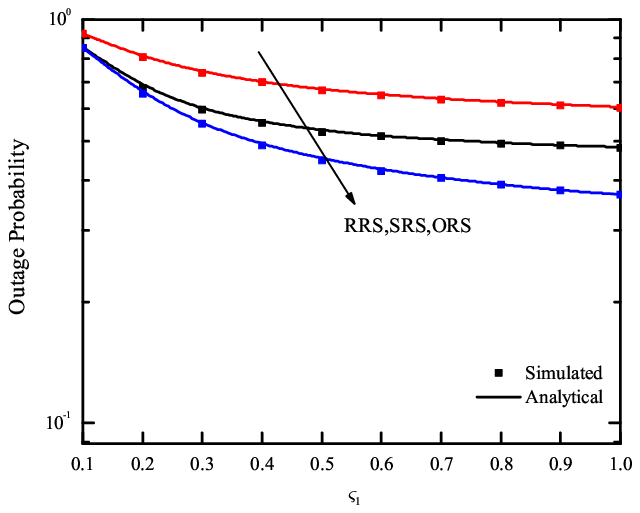}
\end{minipage}
}
\centering
\subfigure[ ]
{\begin{minipage}[t]{0.48\linewidth}
\centering
\includegraphics[width= 2in]{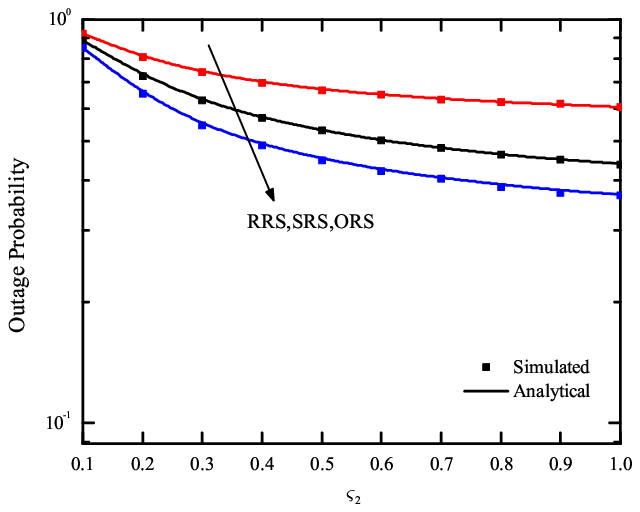}
\end{minipage}
}
\caption{Influence of energy conversion coefficient: (a) OP versus ${\varsigma _1}$; (b) OP versus ${\varsigma _2}$.}
\end{figure}
\begin{figure}[!t]
\centering
\includegraphics [width=3.2in]{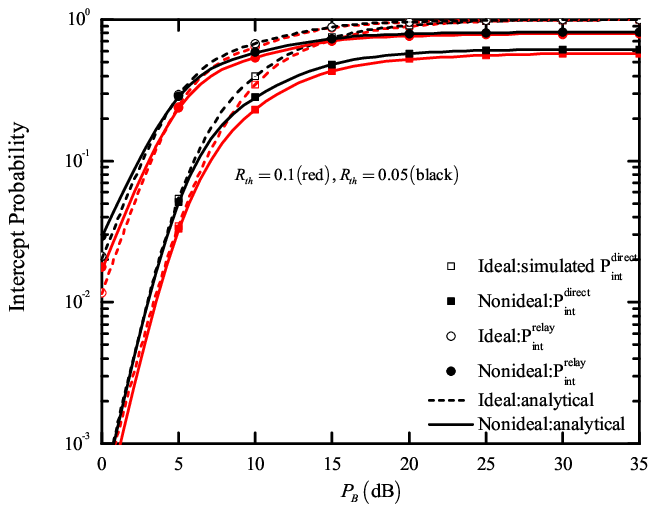}
\caption{{IP versus the transmit power for different threshold rates and transmission schemes.}}
\label{Figure 7}
\end{figure}

Fig. 4 shows the OP versus transmit power ${P_B}$ under three RS strategies for different CEE parameters. In this simulation, we set $M=2$. The simulation results reveal that: 1) when $\sigma _e^2$ is a non-negative constant, the OPs for the three RS schemes are positively correlated with CEE parameters; 2) When $\sigma _e^2 = {\Omega  \mathord{\left/ {\vphantom {\Omega  {\left( {1 + \delta \rho \Omega } \right)}}} \right. \kern-\nulldelimiterspace} {\left( {1 + \delta \rho \Omega } \right)}}$, the OPs decreases with the increase of $\delta $, which means the reliability of the system increases gradually; 3) There are error floors for the OP of the three RS schemes due to fixed non-negative CEEs.

\begin{figure}[!t]
\centering
\includegraphics [width=3.2in]{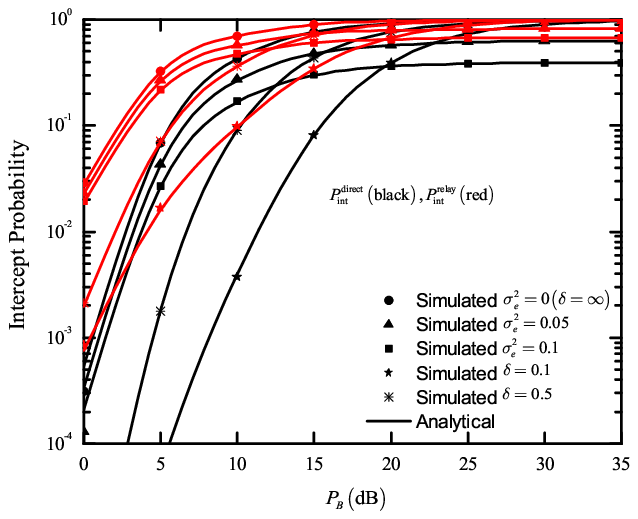}
\caption{{IP versus the transmit power for different CEE parameters.}}
\label{Figure 8}
\end{figure}
\begin{figure}[!t]
\centering
\includegraphics [width=3.2in]{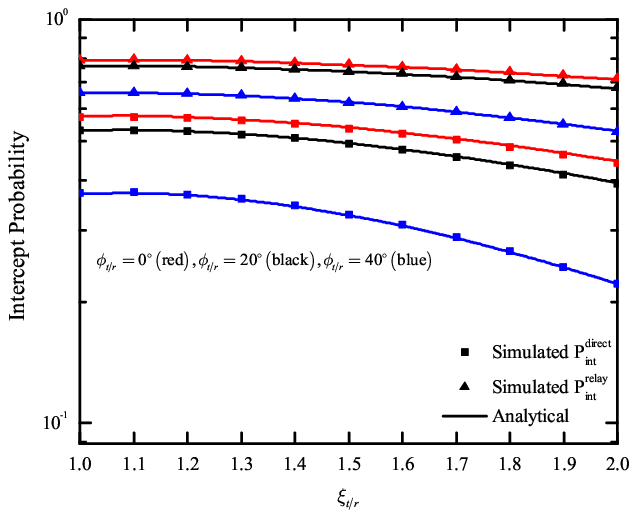}
\caption{{IP versus the TX/RX amplitude for different phase mismatch.}}
\label{Figure 9}
\end{figure}

Fig. 5 shows the OP versus ${P_B}$ for different time allocation efficiencies $\alpha  \in \left\{ {0.2, 0.4, 0.6} \right\}$ with $M = 2$. We have the following observations that: 1) when transmit power ${P_B} \in \left[ {0{\rm{dB}}:24{\rm{dB}}} \right]$, the outage performance of the system becomes stronger as the $\alpha$ gets larger, that is, the reliability of the system increases; 2) when ${P_B} \in \left[ {24{\rm{dB}}:28{\rm{dB}}} \right]$, these simulation results show that OP at $\alpha = 0.6$ is higher than $\alpha = 0.4$; 3) when ${P_B} \in \left[ {28{\rm{dB}}:32{\rm{dB}}} \right]$, the outage performance in the case of $\alpha = 0.6$ is worse than that in $\alpha = 0.2$ and with the increase of  $\alpha = 0.2; 0.4$, the OP
of the system decreases; 4) when ${P_B} \in \left[ {32{\rm{dB}}:40{\rm{dB}}} \right]$, the OP of the system gradually weakens with $\alpha = {0.6; 0.4; 0.2}$, that is, the outage performance of the system gradually improves with the changing order of time allocation efficiency.

Fig. 6 (a) and Fig. 6 (b) plot the OP versus energy conversion coefficients for the three RS schemes. In this simulation, the parameters are set $M = 2$, ${\varsigma _2} = 0.5$ in Fig. 6 (a) and ${\varsigma _1} = 0.5$ in Fig. 6 (b). From Fig. 6 (a) and Fig. 6 (b),  we can see that the OPs under the case of the three RS schemes degrade when the ${\varsigma _1}$ and ${\varsigma _2}$ grow, i.e. the reliability of the system enhances with the increase of the energy conversion coefficients of the system.

\begin{figure}[!t]
\centering
\includegraphics [width=3.2in]{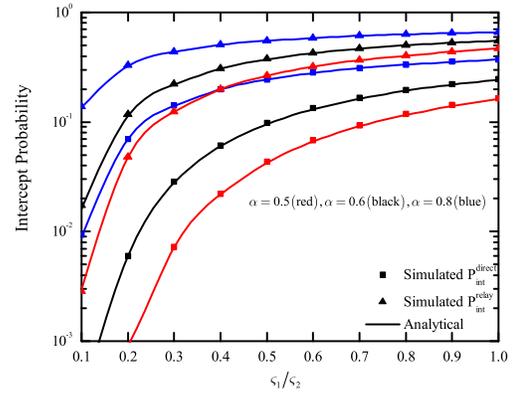}
\caption{{IP versus the energy converse coefficient for different time allocation factors.}}
\label{Figure 10}
\end{figure}

\begin{figure}[!t]
\centering
\includegraphics [width=3.2in]{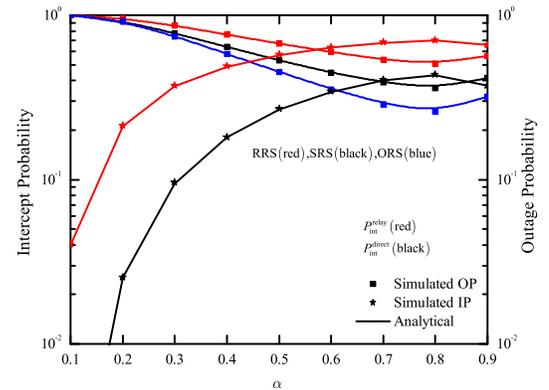}
\caption{{OP and IP versus the time allocation factor for different transmission strategies.}}
\label{Figure 11}
\end{figure}

\subsection{Security Analysis}
Fig. 7 investigates the IP versus transmit power ${P_B}$ for different threshold rates and link schemes under two conditions. The parameters is set as $M = 2$. For different threshold rates, it can be seen that the IP of the system decreases as the ${R_{th}}$ increases, and it can be obtained that the IP in direct link transmission condition is less than that in relay transmission condition. This means that cooperative relay can improve the system performance by shortening the distance between source and destination. We further explore that the IP in the ideal case is smaller than the non-ideal case, that is, the presence of IQI and ICSI in the system will strengthen the security of the system in the high SNR region.

Fig. 8 illustrates the IP versus ${P_B}$ for different CEE $\sigma _e^2$. In this simulation, two CEE cases are considered: 1) $\sigma _e^2$ is a non-negative constant; 2) $\sigma _e^2$ is the function of transmit average SNR. In this simulation, we set $M = 2$. It can be seen that with the aggravation of CEE parameters, the IP of the system becomes smaller. This means that IQI parameters are beneficial to the system IP. Similarly, it can be further concluded that the IP of the system under relay transmission condition is greater than the IP under direct transmission condition.

Fig. 9 shows the IP versus TX/RX amplitude ${\xi _{{t \mathord{\left/ {\vphantom {t r}} \right. \kern-\nulldelimiterspace} r}}}$ for different phase mismatch ${\phi _{{t \mathord{\left/ {\vphantom {t r}} \right. \kern-\nulldelimiterspace} r}}} = \left\{ {{0^ \circ };{{20}^ \circ };{{40}^ \circ }} \right\}$. The simulation results show that the IP under two transmission schemes of the system degrades with the increase of IQI parameter ${\xi _{{t \mathord{\left/ {\vphantom {t r}} \right.
 \kern-\nulldelimiterspace} r}}}$ and phase mismatch ${\phi _{{t \mathord{\left/ {\vphantom {t r}} \right. \kern-\nulldelimiterspace} r}}}$. This means that the IQI existing in this system is negatively correlated with IP under the conditions.

Fig. 10 plots the IP versus energy conversion coefficient ${\varsigma _{{1 \mathord{\left/
 {\vphantom {1 2}} \right.
 \kern-\nulldelimiterspace} 2}}}$ at $S$ and ${R_m}$ for different $\alpha $ with ${P_B} = 5$ dB. From the Fig. 10, we can draw the following conclusions: 1) Under the conditions of IQI and CEEs, the IP for the two transmission schemes of the system is proportional to the time allocation factor; 2) With the increases of energy conversion coefficient, the IP under different time allocation factors gradually increases.

Fig. 11 illustrates that the OP under three RS strategies and IP under two transmission link schemes versus time allocation factor $\alpha$. As can be seen from the simulation, when $\alpha $ increases, the OP of the proposed three RS schemes decreases first and then increases, while IP under two transmission schemes increases first and then decreases in the whole range, which means that there is an optimal value in the process of $\alpha $ gradually increasing. In addition, the optimal solution to balance the reliability and security of the system under consideration can be obtained.

\section{Conclusion and Future Work}
In this paper, we investigate the reliability and security of multi-relay networks in terms of OP and IP in the presence of IQI, ICSI, and nonlinear energy harvesters. To improve the security performance, three RS schemes are considered. For reliability, we analyze the asymptotic behavior in the high SNR regime and discuss the diversity order. For security, we consider two representative cases. Theoretical analysis and experiment results prove that: 1) The OP of the considered system increases as the TX/RX amplitude and phase increases; 2) As the number of relays increases, the system's outage performance becomes better; 3) Different CEE modes have different effects on the system. When the parameter is a non-negative constant, the OP of the system increases as $\sigma _e^2$ increases. When the parameter is a variable, the OP decreases as $\delta $ increases; 4) The performance of the system is proportional to IP and energy conversion coefficient; 5) There is a trade-off between reliability and security, that is, when the performance of the interruption is relaxed, IP can be enhanced, and vice versa; 6) When the system is under the condition of nonideal and the CEEs parameter is a constant, the OP exists error floor.

The work of our paper are focusing on the secure performance of wireless-powered relaying networks affected by IQI, however, our analysis are by no means conclusive and the system performance suffers from other hardware factors, such as, phase noise, amplifier non-linearities and quantization error, etc. To this end, the analytical method of our work can be extended to the above the hardware imperfections. In fact, the ICSI is caused not only by CCE at receiver, but also by feedback delay at the transmitter. Our analysis can be extended to investigate the secure performance of multi-antenna cooperative systems. The above exciting extensions would be done as our future work.

\numberwithin{equation}{section}
\section*{Appendix~A: Proof of Theorem 1} 
\renewcommand{\theequation}{A.\arabic{equation}}
\setcounter{equation}{0}
According to the definition of OP and (8), the following expression can be obtained as:
\begin{align}
\label{1}\nonumber
P_{out}^{RRS} \!&=\! \Pr\! \left\{ {\min \left(\! {{C_{S\!{R\!_m}}},{C_{{R\!_m}\!D}}} \!\right)\! <\! {R\!_{th}}}\! \right\}\\
& = 1\! -\! \underbrace {\Pr \!\left\{ {{C\!_{S\!{R\!_m}}} \!>\! {R\!_{th}}}\! \right\}}_{{I\!_1}}\underbrace {\Pr \left\{ {{C\!_{{R\!_m}\!D}} \!>\! {R\!_{th}}} \!\right\}}_{{I\!_2}}.
\end{align}

Substituting (8) into (A1), set $\varepsilon  = {2^{\frac{{2{R_{th}}}}{{1 - \alpha }}}}$ and the ${I_1}$ can be rewritten as:
\begin{align}\nonumber
{I_1} &\!\!=\!\! \Pr \!\!\left\{\!\!\! {\frac{{{{\left|\! {{{\hat h}\!_{S\!{R\!_m}}}} \!\right|}^2}{\rho \!_{S\!{R\!_m}}}\!{p\!_{S\!R\!_m}}}}{{\sigma _{{e\!_{S\!{R\!_m}}}}^2\!\!{\rho \!_{S\!{R\!_m}}}\!{p\!_{S\!R\!_m}} \!\!\!+\!\! {{\left| \!{{{\hat h}\!_{S\!{R\!_m}}}} \!\right|}^2}\!\!{\rho \!_{S\!{R\!_m}}}\!{q\!_{S\!R\!_m}} \!\!\!+\!\! \sigma _{{e\!_{S\!{R\!_m}}}}^2\!\!{\rho \!_{S\!{R\!_m}}}\!\!{q\!_{S\!R\!_m}} \!\!\!+\!\! {g\!_{S\!R\!_m}}}} \!\!>\!\! \varepsilon }\!\!\! \right\}\\
 &= {M_1} + {M_2},
\end{align}
where
\begin{align}\nonumber
{M_1} & \!\!=\! \!\Pr \!\!\left\{ \!\!{\frac{{{g\!_{S\!R\!_m}}\!{N\!_{S\!{R\!_m}}}\!\varepsilon }}{{{C\!_1}\!{{\left|\! {{{\hat h}\!_{S\!{R\!_m}}}} \!\right|}^2} \!\!-\!\! {C\!_2}}} \!\!<\!\! {{\left| \!{{h\!_{B\!S}}} \!\right|}^2}\!\! \le \!\!{E\!_1},{{\left|\! {{{\hat h}\!_{S\!{R\!_m}}}} \!\right|}^2} \!\!\ge\!\! \frac{{{g\!_{S\!R\!_m}}\!{N\!_{S\!{R\!_m}}}\!\varepsilon }}{{{C\!_1}{E\!_1}}} \!\!+\!\! \frac{{{C\!_2}}}{{{C\!_1}}}} \!\!\right\},
\end{align}
and
\begin{align}\nonumber
{M_2} &= \Pr \left\{ {{{\left| {{{\hat h}_{S{R_m}}}} \right|}^2} > {\Theta _1},{{\left| {{h_{BS}}} \right|}^2} > {E_1}} \right\}\\
& = {e^{ - {\lambda _{S{R_m}}}{\Theta _1} - {\lambda _{BS}}{E_1}}},
\end{align}
in there, ${T_2} = \frac{{{g_{SR_m}}{N_{S{R_m}}}\varepsilon }}{{{C_1}{{\left| {{{\hat h}_{S{R_m}}}} \right|}^2} - {C_2}}}$.

Substituting the PDF and CDF of Rayleigh fading into (A3), the following formula can be obtained by further calculation as:
\begin{equation}
{M_1} =  - {\lambda _{S{R_m}}}\left( {{e^{ - {\lambda _{BS}}{E_1}}}{\varphi _1} - {\varphi _2}} \right),
\end{equation}
\begin{equation}
{\varphi _1} = \int_{{T_1}}^\infty  {{e^{ - {\lambda _{S{R_m}}}y}}dy}  = \frac{1}{{{\lambda _{S{R_m}}}}}{e^{ - {\lambda _{S{R_m}}}{T_1}}},
\end{equation}
according to the formula (3.324.1) in \cite{50} and the following expression (A7) of Gaussian-Chebyshev quadrature \cite{51}, the ${\varphi _2}$ of the top of next page can be obtained as
\begin{equation}
\int_0^\Lambda  {g\left( x \right)dx \approx \frac{{\pi \Lambda }}{{2Y}}} \sum\limits_{l = 0}^Y {g\left( {\frac{{\Lambda \left( {{\delta _l} + 1} \right)}}{2}} \right)} \sqrt {1 - \delta _l^2} ,
\end{equation}
\begin{figure*}[!t]
\normalsize
\begin{align}\nonumber
{\varphi _2} &= \int_{{T_1}}^\infty  {{e^{ - {\lambda _{S{R_m}}}y}}{e^{ - \frac{{{\lambda _{BS}}{g_{SR_m}}{N_{S{R_m}}}\varepsilon }}{{{C_1}x - {C_2}}}}}dy} \\
 &= \frac{1}{{{C_1}}}{e^{ - \frac{{{\lambda _{S{R_m}}}{C_2}}}{{{C_1}}}}}\left[ {\sqrt {\frac{{{\beta _1}}}{{{\gamma _1}}}} {{\rm K}_1}\left( {\sqrt {{\beta _1}{\gamma _1}} } \right) - } \right.\left. {\frac{{\pi {\Lambda _1}}}{{2{Y_1}}}\sum\limits_{{l_1} = 0}^{{Y_1}} {{e^{ - \frac{{2{\lambda _{BS}}{g_{SR_m}}{N_{S{R_m}}}\varepsilon }}{{{\Lambda _1}\!\left( {{\delta _{{l_1}}} + 1} \right)}} - \frac{{{\lambda _{S{R_m}}}{\Lambda _1}\left( {{\delta _{{l_1}}} + 1} \right)}}{{2{C_1}}}}}} \sqrt {1 - \delta _{{l_1}}^2} } \right],
\end{align}
\hrulefill \vspace*{0pt}
\end{figure*}
\noindent Substituting (A6) and (A8) into (A5), the ${M_1}$ can be derived. Then substituting (A4) and (A5) into (A2), the ${I_1}$ can be obtained.

Substituting (8) into (A1), the ${I_2}$ of following expression can be rewritten as
\begin{align}
{I_2} = {M_3} + {M_4}
\end{align}

Similar to the calculation of ${I_1}$, the ${M_3}$ in the top of next page and ${M_4}$ can be obtained as follows
\begin{figure*}[!t]
\normalsize
\begin{align}
{M_3} \!=\!  -\! {e^{ -\! {\lambda \!_{B\!{R\!_m}}}{E\!_2} \!-\! {\lambda _{{R\!_m}\!D}}{T\!_3}}} \!+\! \frac{{{\lambda _{{R\!_m}\!D}}}}{{{C\!_3}}}{e^{ -\! \frac{{{\lambda _{{R\!_m}\!D}}{C\!_4}}}{{{C\!_3}}}}}\left[ \!{\sqrt {\!\frac{{{\beta \!_2}}}{{{\gamma \!_2}}}} {{\rm K}\!_1}\!\sqrt {\!{\beta \!_2}{\gamma \!_2}} \! - \!\frac{{\pi {\Lambda _2}}}{{2{Y\!_2}}}\!\sum\limits_{{l_2} = 0}^{{Y\!_2}} {{e^{ -\! \frac{{2{\lambda _{B\!{R\!_m}}}{g\!_{R\!_m\!D}}{N\!_{{R\!_m}\!D}}\varepsilon }}{{{\Lambda \!_2}\left(\! {{\delta _{{l_2}}}\! +\! 1} \!\right)}}\!-\! \frac{{{\lambda _{{R_m}D}}{\Lambda _2}\left( {{\delta _{{l_2}}} \!+\! 1} \!\right)}}{{2{C_3}}}}}}\! \sqrt {1 \!-\! \delta _{{l_2}}^2} }\! \right],
\end{align}
\hrulefill \vspace*{0pt}
\end{figure*}
and
\begin{equation}
{M_4} = {e^{ - {\lambda _{{R_m}D}}{\Theta _2} - {\lambda _{B{R_m}}}{E_2}}},
\end{equation}
put (A10) and (A11) into (A9), the ${I_2}$ can be derived.

Substituting the expressions of ${I_1}$ and ${I_2}$, the (13) can be obtained.

\numberwithin{equation}{section}
\section*{Appendix~B: Proof of Theorem 2} 
\renewcommand{\theequation}{B.\arabic{equation}}
\setcounter{equation}{0}

For SRS strategy, substituting (21) into (11), the following expression can be obtained as
\begin{align}\nonumber
P_{out}^{SRS}& = \Pr \left\{ {\min \left( {{C_{S{R_a}}},{C_{{R_a}D}}} \right) < {R_{th}}} \right\}\\
& = 1 - \underbrace {\Pr \left\{ {{C_{S{R_a}}} > {R_{th}}} \right\}}_{{I_3}}\underbrace {\Pr \left\{ {{C_{{R_a}D}} > {R_{th}}} \right\}}_{{I_4}},
\end{align}
the CDF and PDF of ${\left| {{{\hat h}_{S{R_a}}}} \right|^2}$ can be written as
\begin{equation}
{F_{{{\left| {{{\hat h}_{{R_a}D}}} \right|}^2}}}\left( y \right) = {\left[ {1 - {e^{ - {\lambda _{{R_a}D}}y}}} \right]^M},
\end{equation}
\begin{equation}
{f_{{{\left|\! {{{\hat h}\!_{S\!{R\!_a}}}} \!\right|}^2}}}\!\left( \!y \right) \!=\! M{\lambda \!_{S\!{R\!_a}}}\!\sum\limits_{s\! =\! 0}^{M\! - \!1} {\left(\! {\begin{array}{*{20}{c}}
{M\! - \!1}\\
s
\end{array}} \!\right){{\left(\! { - \!1} \!\right)}^s}{e^{ -\! {\lambda \!_{S\!{R\!_a}}}\!\left(\! {s \!+\! 1} \!\right)y}}} .
\end{equation}

Similar to the calculation process of Appendix A, the ${I_3}$ and ${I_4}$ can be expressed as
\begin{align}
{I_3}  = {M_5} + {M_6},
\end{align}
where
\begin{align}\nonumber
{M_5} &=\! \Pr\! \left\{\! {{{\left|\! {{h\!_{B\!S}}} \!\right|}^2}\left(\! {{C\!_5}{{\left|\! {{{\hat h}\!_{S\!{R\!_a}}}}\! \right|}^2}\! -\! {C\!_6}} \!\right) \!>\! {g\!_{S\!R\!_m}}{N\!_{S\!{R\!_a}}}\!\varepsilon ,{{\left|\! {{h\!_{B\!S}}} \!\right|}^2} \!\le \!{E\!_1}}\! \right\}\\
& = \Xi \left( {{e^{ - {\lambda _{BS}}{E_1}}}{\varphi _5} - {\varphi _6}} \right),
\end{align}
\begin{align}
{\varphi _5} = \int_{{T_5}}^\infty  {{e^{ - {\lambda _{S{R_a}}}\left( {s + 1} \right)x}}dx}  = \frac{1}{{{\lambda _{S{R_a}}}\left( {s + 1} \right)}}{e^{ - {\lambda _{S{R_a}}}\left( {s + 1} \right){T_5}}},
\end{align}
the ${\varphi _6}$ in the top of next page,
\begin{figure*}[!t]
\normalsize
\begin{align}\nonumber
{\varphi _6}& = \int_{{T_5}}^\infty  {{e^{ - {\lambda _{S{R_a}}}\left( {s + 1} \right)x - {\lambda _{BS}}{T_6}}}dx} \\
& = \frac{1}{{{C_5}}}{e^{ - \frac{{{\lambda _{S{R_a}}}\left( {s + 1} \right){C_6}}}{{{C_5}}}}}\left[ {\sqrt {\frac{{{\beta _3}}}{{{\gamma _3}}}} {{\rm K}_1}\left( {\sqrt {{\beta _3}{\gamma _3}} } \right) - \frac{{\pi {\Lambda _3}}}{{2{Y_3}}}\sum\limits_{{l_3} = 0}^{{Y_3}} {{e^{ - \frac{{2{\lambda _{BS}}{g\!_{S\!R\!_m}}{N_{S{R_a}}}\varepsilon }}{{{\Lambda _3}\left( {{\delta _{{l_3}}} + 1} \right)}} - \frac{{{\lambda _{S{R_a}}}\left( {s + 1} \right){\Lambda _3}\left( {{\delta _{{l_3}}} + 1} \right)}}{{2{C_5}}}}}\sqrt {1 - \delta _{{l_3}}^2} } } \right],
\end{align}
\hrulefill \vspace*{0pt}
\end{figure*}
and
\begin{align}\nonumber
{M_6} &= \Pr \left\{ {{{\left| {{{\hat h}_{S{R_a}}}} \right|}^2} > {\Theta _3},{{\left| {{h_{BS}}} \right|}^2} > {E_1}} \right\}\\
& = \left[ {1 - {{\left( {1 - {e^{ - {\lambda _{S{R_a}}}{\Theta _3}}}} \right)}^M}} \right]{e^{ - {\lambda _{BS}}{E_1}}},
\end{align}
putting (B6) and (B7) into (B5), the ${M_5}$ can be obtained; substituting (B5) and (B8) into (B4), the ${I_3}$ can be derived.

Then, substituting (8) into (B1), the following formula can be expressed as
\begin{align}
{I_4} = {M_7} + {M_8},
\end{align}
where
\begin{align}\nonumber
{M_7} &=\! \Pr\! \left\{\! {{{\left|\! {{h\!_{B\!{R\!_a}}}} \!\right|}^2}\!\left(\! {{C\!_7}{{\left|\! {{{\hat h}\!_{{R\!_a}\!D}}} \!\right|}^2}\! -\! {C\!_8}} \!\right) \!\!>\!\! {g\!_{R\!_m\!D}}\!{N\!_{{R\!_a}\!D}}\!\varepsilon ,{{\left| {{h\!_{B\!{R\!_a}}}} \!\right|}^2}\! \le\! {E\!_2}} \!\right\}\\
& =  - {\lambda _{{R_a}D}}\left( {{e^{ - {\lambda _{B{R_a}}}{E_2}}}{\varphi _7} - {\varphi _8}} \right),
\end{align}
\begin{equation}
{\varphi _7} = \int_{{T_7}}^\infty  {{e^{ - {\lambda _{{R_a}D}}y}}dy}  = \frac{1}{{{\lambda _{{R_a}D}}}}{e^{ - {\lambda _{{R_a}D}}{T_7}}},
\end{equation}
the ${\varphi _8}$ in the top of next page and $M_8$ is obtained as following
\begin{figure*}[!t]
\normalsize
\begin{align}
{\varphi _8} = \frac{1}{{{C_7}}}{e^{ - \frac{{{\lambda _{{R_a}D}}{C_8}}}{{{C_7}}}}}\left[ {\sqrt {\frac{{{\beta _4}}}{{{\gamma _4}}}} {{\rm K}_1}\sqrt {{\beta _4}{\gamma _4}}  - \frac{{\pi {\Lambda _4}}}{{2{Y_4}}}\sum\limits_{{l_4} = 0}^{{Y_4}} {{e^{ - \frac{{2{\lambda _{B{R_a}}}{g_{R_mD}}{N_{{R_a}D}}\varepsilon }}{{{\Lambda _4}\left( {{\delta _{{l_4}}} + 1} \right)}} - \frac{{{\lambda _{{R_a}D}}{\Lambda _4}\left( {{\delta _{{l_4}}} + 1} \right)}}{{2{C_7}}}}}} \sqrt {1 - \delta _{{l_4}}^2} } \right],
\end{align}
\hrulefill \vspace*{0pt}
\end{figure*}
\begin{align}\nonumber
{M_8} &= \Pr \left\{ {{{\left| {{{\hat h}_{{R_a}D}}} \right|}^2} > {\Theta _4},{{\left| {{h_{B{R_a}}}} \right|}^2} > {E_2}} \right\}\\
 &= {e^{ - {\lambda _{{R_a}D}}{\Theta _4} - {\lambda _{B{R_a}}}{E_2}}},
\end{align}
putting (B11) and (B12) into (B10), the ${M_7}$ can be obtained; substituting (B10) and (B13) into (B9), the ${I_4}$ can be derived.

Substituting ${I_3}$ and ${I_4}$ into (B1), the (22) can be obtained.

\numberwithin{equation}{section}
\section*{Appendix~C: Proof of Theorem 3} 
\renewcommand{\theequation}{C.\arabic{equation}}
\setcounter{equation}{0}

According to the definition of (11) and (25), the following expression for ORS strategy can be obtained as
\begin{align}\nonumber
P_{out}^{ORS} &= \Pr \left\{ {{C_{{R_{{m^ * }}}}} < {R_{th}}} \right\}\\\nonumber
& = \Pr \left\{ {\mathop {\max }\limits_{1 \le m \le M} \min \left\{ {{\gamma _{S{R_m}}},{\gamma _{{R_m}D}}} \right\} < \varepsilon } \right\}\\
& = \prod\limits_{m = 1}^M {\left( {1 - {I_1}{I_2}} \right)},
\end{align}
put ${I_1}$ and ${I_2}$ of Appendix A into (C1), the (26) can be obtained.

\numberwithin{equation}{section}
\section*{Appendix~D: Proof of Theorem 4} 
\renewcommand{\theequation}{D.\arabic{equation}}
\setcounter{equation}{0}
Substituting (8) into (30), the following expression can be obtained as
\begin{align}\nonumber
P_{{\rm{int}}}^{{\rm{direct}}} &\!\buildrel \Delta \over =\! \Pr\! \left\{\! {{C\!_{S\!E}} \!>\! {R\!_{th}}} \!\right\}\\
& = {M_9} + {M_{10}}.
\end{align}

Similar to the Appendix A, the ${M_{9}}$ and ${M_{10}}$ can be expressed as
\begin{align}\nonumber
{M\!_9} \!&=\! \Pr\! \left\{\! {\underbrace {\frac{{{g_{SE}}{N\!_{S\!E}}\varepsilon }}{{{C\!_9}{{\left| {{{\hat h}\!_{S\!E}}} \right|}^2} \!\!-\!\! {C\!_{10}}}}}_{{T\!_9}} \!\!<\!\! {{\left| {{{\hat h}\!_{B\!S}}} \right|}^2} \!\!\le\!\! {E\!_1},{{\left| {{{\hat h}\!_{S\!E}}} \right|}^2} \!\!\ge\!\! \underbrace {\frac{{{g\!_{S\!E}}{N\!_{S\!E}}\varepsilon }}{{{C\!_9}{E\!_1}}} \!\!+\!\! \frac{{{C\!_{10}}}}{{{C\!_9}}}}_{{T\!_{10}}}} \!\right\}\\
& =  - {\lambda \!_{S\!E}}\left(\! {{e^{ -\! {\lambda \!_{B\!S}}{E\!_1}}}{\varphi \!_9}\! -\! {\varphi \!_{10}}} \!\right),
\end{align}
\begin{align}\nonumber
{M_{10}} &= \Pr \left\{ {{{\left| {{{\hat h}_{SE}}} \right|}^2} > {\Theta _5},{{\left| {{{\hat h}_{BS}}} \right|}^2} > {E_1}} \right\}\\
 &= {e^{ - {\lambda _{SE}}{\Theta _5} - {\lambda _{BS}}{E_1}}},
\end{align}
where
\begin{align}
{\varphi _9} = \int_{{T_9}}^\infty  {{e^{ - {\lambda _{SE}}y}}dy}  = \frac{1}{{{\lambda _{SE}}}}{e^{ - {\lambda _{SE}}{T_9}}},
\end{align}
and ${\varphi _{10}}$ on the top of next page.
\begin{figure*}[!t]
\normalsize
\begin{align}
{\varphi _{10}} = \frac{1}{{{C_9}}}{e^{ - \frac{{{\lambda _{SE}}{C_{10}}}}{{{C_9}}}}}\left[ {2\sqrt {\frac{{{\beta _5}}}{{{\gamma _5}}}} {{\rm K}_1}\left( {2\sqrt {{\beta _5}{\gamma _5}} } \right) - \frac{{\pi {\Lambda _5}}}{{2{Y_5}}}\sum\limits_{{l_5} = 0}^{{Y_5}} {{e^{ - \frac{{{\gamma _5}{\Lambda _5}\left( {{\delta _{{l_5}}} + 1} \right)}}{2} - \frac{{2{\beta _5}}}{{{\Lambda _5}\left( {{\delta _{{l_5}}} + 1} \right)}}}}\sqrt {1 - \delta _{{l_5}}^2} } } \right],
\end{align}
\hrulefill \vspace*{0pt}
\end{figure*}
Putting (D4) and (D5) into (D2), the ${M_{9}}$ can be derived; then, substituting ${M_{9}}$ and ${M_{10}}$ into (D1), the (31) can be obtained.


\numberwithin{equation}{section}
\section*{Appendix~E: Proof of Theorem 5} 
\renewcommand{\theequation}{E.\arabic{equation}}
\setcounter{equation}{0}

Substituting (8) into (26), the following expression can be obtained as
\begin{align}\nonumber
P_{{\rm{int}}}^{{\rm{relay}}} & \buildrel \Delta \over = \Pr \left\{ {{C_{{R_c}E}} > {R_{th}}} \right\}\\
& = {M_{11}} + {M_{12}}.
\end{align}

Similar to the Appendix A, the ${M_{11}}$ and ${M_{12}}$ can be expressed as
\begin{align}\nonumber
{M_{11}} \!\!&=\!\! \Pr\!\! \left\{\!\! {\underbrace {\frac{{{g\!_{R\!_c\!E}}{N\!_{{R\!_c}\!D}}\!\varepsilon }}{{{C\!_9}\!{{\left|\! {{{\hat h}\!_{{R\!_c}\!D}}} \!\right|}^2} \!\!-\!\! {C\!_{10}}}}}_{{T\!_{10}}}\! \!<\! \!{{\left|\! {{h\!_{B\!{R\!_c}}}} \!\right|}^2}\! \le\! {E\!_2},{{\left|\! {{{\hat h}\!_{{R\!_c}\!E}}} \!\right|}^2} \!\!\ge\!\! \underbrace {\frac{{{g\!_{R\!_c\!E}}{N\!_{{R\!_c}\!E}}\!\varepsilon }}{{{C\!_9}\!{E\!_2}}} \!\!+\!\! \frac{{{C\!_{10}}}}{{{C\!_9}}}}_{{T\!_9}}} \!\!\right\}\\
& =  - {\lambda _{{R_c}E}}\left( {{e^{ - {\lambda _{B{R_c}}}{E_2}}}{\varphi _9} - {\varphi _{10}}} \right),
\end{align}
\begin{align}\nonumber
{M_{12}} &= \Pr \left\{ {{{\left| {{{\hat h}_{{R_c}E}}} \right|}^2} > {\Theta _5},{{\left| {{h_{B{R_c}}}} \right|}^2} > {E_2}} \right\}\\
& = {e^{ - {\lambda _{{R_c}E}}{\Theta _5} - {\lambda _{B{R_c}}}{E_2}}},
\end{align}
where
\begin{align}
{\varphi _9} = \int_{{T_9}}^\infty  {{e^{ - {\lambda _{{R_c}E}}y}}dy}  = \frac{1}{{{\lambda _{{R_c}E}}}}{e^{ - {\lambda _{{R_c}E}}{T_9}}},
\end{align}
and ${\varphi _{10}}$ in the top of next page.
\begin{figure*}[!t]
\normalsize
\begin{align}
{\varphi _{10}} = \frac{1}{{{C_9}}}{e^{ - \frac{{{\lambda _{{R_c}E}}{C_{10}}}}{{{C_9}}}}}\left[ {\sqrt {\frac{{{\beta _5}}}{{{\gamma _5}}}} {{\rm K}_1}\sqrt {{\beta _5}{\gamma _5}}  - \frac{{\pi {\Lambda _5}}}{{2{Y_5}}}\sum\limits_{{l_5} = 0}^{{Y_5}} {{e^{ - \frac{{2{\lambda _{B{R_c}}}{g_{R_cE}}{N_{{R_c}E}}\varepsilon }}{{{\Lambda _5}\left( {{\delta _{{l_5}}} + 1} \right)}} - \frac{{{\lambda _{{R_c}E}}{\Lambda _5}\left( {{\delta _{{l_5}}} + 1} \right)}}{{2{C_9}}}}}} \sqrt {1 - \delta _{{l_5}}^2} } \right],
\end{align}
\hrulefill \vspace*{0pt}
\end{figure*}

Putting (E4) and (E5) into (E2), the ${M_{11}}$ can be derived; then, substituting ${M_{11}}$ and ${M_{12}}$ into (E1), the (32) can be obtained.

\bibliographystyle{IEEEtran}
\bibliography{myreference}

\end{document}